\documentclass{article}
\usepackage{amsmath,amssymb,amsthm}
\usepackage{hyperref}

\newtheorem{theorem}{Theorem}
\newtheorem{definition}{Definition}
\newtheorem{proposition}{Proposition}
\newtheorem{lemma}{Lemma}

\newcommand{\F}{\mathbb{F}_2}

\newcommand{\I}{\mathbb{I}}

\newcommand{\twobytwo}[4]{\begin{pmatrix} #1 & #2 \\ #3 & #4\end{pmatrix}}
\newcommand{\twobyone}[2]{\begin{pmatrix} #1 \\ #2\end{pmatrix}}

\newcommand{\A}{\mathbf{A}}
\newcommand{\B}{\mathbf{B}}
\newcommand{\E}{\mathbf{E}}
\newcommand{\Bp}{\mathbf{B^p}}
\newcommand{\Ap}{\mathbf{A^p}}
\newcommand{\Bu}{\mathbf{B^u}}
\newcommand{\Au}{\mathbf{A^u}}
\newcommand{\Bb}{\mathbf{B^\beta}}
\newcommand{\Ab}{\mathbf{A^\beta}}

\newcommand{\Gc}{\mathcal{G}^\beta}
\newcommand{\Ac}{\mathcal{A}^\beta}
\newcommand{\Bc}{\mathcal{B}^\beta}
\newcommand{\Gcb}{\bar{\mathcal{G}}^\beta}
\newcommand{\Acb}{\bar{\mathcal{A}}^\beta}
\newcommand{\Bcb}{\bar{\mathcal{B}}^\beta}

\newcommand{\Acu}{\underline{\mathcal{A}}}
\newcommand{\Bcu}{\underline{\mathcal{B}}}
\newcommand{\Acp}{\mathcal{A}^{o.p.}}
\newcommand{\Bcp}{\mathcal{B}^{o.p.}}
\newcommand{\X}[1]{\hat{X}^{(#1)}}
\newcommand{\W}[1]{\hat{W}^{(#1)}}

\newcommand{\dax}{d_A}
\newcommand{\dbx}{d_B}

\newcommand{\subm}[3]{{#1}_{#2,#3}}
\newcommand{\esubm}[3]{\tilde{#1}_{#2,#3}}
\newcommand{\submi}[3]{{#1}_{#2,#3}^{-1}}
\newcommand{\esubmi}[3]{\tilde{#1}_{#2,#3}^{-1}}
\newcommand{\submt}[3]{{#1}_{#2,#3}^\top}
\newcommand{\esubmt}[3]{\tilde{#1}_{#2,#3}^\top}
\newcommand{\submti}[3]{{#1}_{#2,#3}^{-\top}}
\newcommand{\esubmti}[3]{\tilde{#1}_{#2,#3}^{-\top}}

\newcommand{\rs}{\mathbf{RS}}
\newcommand{\cs}{\mathbf{CS}}

\begin{document}
\title{The Find Rows and Columns and Decode algorithm for quantum expander codes}
\author{Dimiter Ostrev\thanks{Interdisciplinary Centre for Security, Reliability and Trust, University of Luxembourg, L-4364 Esch-sur-Alzette, Luxembourg}}
\date{}
\maketitle

\begin{abstract}
A new adaptation of the Find Erasures and Decode algorithm from classical to quantum expander codes is presented. It runs in linear time and is parallelizable to logarithmic depth. Compared to Small Set Flip and Small Set Find, the new algorithm avoids the overhead of considering the subsets of stabilizer generators, requires less expansion and corrects more errors. 
\end{abstract}

\section{Introduction}

A lossless expander is a sparse bipartite graph in which all small sets of vertices have almost as many neighbors as edges. Such graphs can be used to construct classical linear codes that can be decoded in linear time or by a circuit with linear width and logarithmic depth \cite{sipser1996expander}. The decoder, often called the Flip Algorithm, searches for variables that can be flipped to decrease the number of unsatisfied constraints. An alternative approach, called Find Erasures and Decode, first searches for a small superset of the error positions, then uses this to correct the error \cite{viderman2013lineartime}. Find Erasures and Decode has similar time complexity but requires less expansion and corrects more errors than the Flip Algorithm. A third approach combines the first two with some additional techniques \cite{chen2023improveddecoding}. If the graph is a good expander, then the combined approach corrects more errors than Find Erasures and Decode alone; this, however, comes at the cost of applying the first two algorithms repeatedly until the correct value of some parameters is guessed. 

Expander graphs have also been used to construct quantum LDPC codes: Calderbank-Shor-Steane codes \cite{calderbank1996good,steane1996multiple} with sparse parity check matrices. Quantum expander codes \cite{leverrier2015quantum} are obtained from the hypergraph product \cite{tillich2014quantum} of a two-sided lossless expander with itself. Decoding these codes presents a new challenge: there are small errors on which the classical Flip Algorithm cannot make progress. A way around this is to consider flipping more than one qubit at a time. This leads to the Small Set Flip algorithm, which can correct adversarial errors up to a constant fraction of the code distance \cite{leverrier2015quantum} or random errors occurring at a constant rate \cite{fawzi2018efficientdecoding}. Small Set Flip is also parallelizable \cite{fawzi2018constant}. 

The distance of hypergraph product codes scales with the square root of the blocklength \cite{tillich2014quantum}. Subsequently, lifted product codes \cite{panteleev2021asymptotically} and quantum Tanner codes \cite{leverrier2022quantum} achieved distance growing linearly with the blocklength, and the Small Set Flip algorithm was adapted to these quantum LDPC families \cite{dinur2023goodquantum,gu2023efficient,leverrier2023efficient,leverrier2023decoding}. 

Lifted product codes and quantum Tanner codes use spectral expanders in their construction. Asymptotically good quantum LDPC codes can also be constructed from the balanced product of lossless expanders and the Small Set Flip algorithm can correct a constant fraction of adversarial errors \cite{lin2022good}. Lossless expanders with all the properties required for this construction were not known to exist at first, but were later discovered \cite{hsieh2025explicitlosslessvertexexpanders}. 

The classical Find Erasures and Decode algorithm has also been adapted to quantum expander codes; the resulting decoder is called Small Set Find plus Erasure Decoding \cite{krishna2024vidermansalgorithmforquantum}. Unlike the classical case, the Erasure Decoding stage of this algorithm takes time $O(n^{3/2})$ for blocklength $n$. A different approach, called ReShape \cite{quintavalle2022reshape}, is to reduce the decoding of quantum expander codes to the corresponding problem for the underlying classical expander codes and then to apply the Find Erasures and Decode algorithm. However, the time complexity of this reduction scales quadratically with the blocklength. 

This article introduces a new adaptation of the classical Find Erasures and Decode algorithm to quantum expander codes. The new approach, called Find Rows and Columns and Decode, is faster, requires less expansion and corrects more errors than previous algorithms. In addition, the analysis extends naturally to irregular graphs. 

The Find Rows and Columns and Decode algorithm runs in linear time and is parallelizable to logarithmic depth. It is faster than the ReShape reduction by a factor linear in the blocklength and faster than the second stage of Small Set Find plus Erasure Decoding by a factor that scales as the square root of the blocklength. Compared to Small Set Flip and Small Set Find, the Find Rows and Columns and Decode algorithm is faster by an $\Omega(2^d)$ factor, where $d$ is the maximum degree of the graph; this is because the subsets of stabilizer generators are not used. 

Next, consider the expansion requirements and decoding radius of the different algorithms. The ReShape reduction combined with the classical Find Erasures and Decode algorithm achieves the same results as Find Rows and Columns and Decode, but the time complexity of ReShape is quadratic rather than linear. The comparison of Find Rows and Columns and Decode with Small Set Flip and Small Set Find is given in detail below. 

The level of expansion of a graph can be quantified by a parameter $\delta$: small sets are guaranteed to have a number of neighbors more than a $(1-\delta)$ fraction of the number of their edges. In this parametrization, smaller $\delta$ corresponds to more expansion. The analysis of Find Rows and Columns and Decode requires lossless expansion parameter $\delta<1/3$. One version of Small Set Flip needs $\delta <1/6$ \cite{leverrier2015quantum} and another version $\delta <1/8$ \cite{fawzi2018efficientdecoding,fawzi2018constant}. Small Set Find requires $\delta<1/10$. The smaller the lossless expansion parameter, the smaller the relative weight of sets for which this level of expansion can be guaranteed \cite[Section 8.4]{richardson2008modern}. Moreover, lossless expanders with very small $\delta$ must have correspondingly very high degrees of at least $1/\delta$ \cite[Fact 12]{chen2023improveddecoding}, which should be avoided if possible in the context of quantum LDPC codes. 

The Find Rows and Columns and Decode algorithm corrects more errors than Small Set Flip and Small Set Find plus Erasure Decoding. Consider the example in \cite[Section 1.1.1]{krishna2024vidermansalgorithmforquantum} of a regular two-sided expander with left to right degree ratio $1/2$, with lossless expansion parameter $\delta=1/20$ and with $s$ the size of sets up to which this level of expansion can be guaranteed. Then, Small Set Find plus Erasure Decoding corrects $0.062 s$ errors, Small Set Flip corrects $0.058 s$ errors, and Find Rows and Columns and Decode corrects $0.944 s$ errors, more than an order of magnitude higher. 

The results here answer affirmatively some questions in the quantum LDPC literature. Reference \cite[Section 1.3]{krishna2024vidermansalgorithmforquantum} asked if it is possible to have a quantum version of the Find Erasures and Decode algorithm with improved parameters. This is achieved here. Reference \cite[Section 4]{gu2023efficient} asked if it is possible to avoid the overhead of considering the subsets of stabilizer generators and if it is possible to reduce the maximum degrees of expander graphs used in quantum LDPC constructions. This article shows that in the case of quantum expander codes, the overhead can be avoided and the degrees can be lowered.

After some preliminaries in Section \ref{sec:preliminaries}, the Find Rows and Columns and Decode algorithm and its analysis are presented in Section \ref{sec:find_rows_and_columns_and_decode}. Section \ref{sec:conclusion} concludes and gives some directions for future work. 

\section{Preliminaries}\label{sec:preliminaries}

Section \ref{sec:notation} introduces some notation. Section \ref{sec:lossless_expanders} states the definition of two-sided lossless expanders. Section \ref{sec:graph_neighborhoods} introduces a convenient generalization of graph neighborhoods, and Section \ref{sec:sequences_of_subgraphs} uses these to reformulate the first stage of the classical Find Erasures and Decode algorithm as an increasing sequence of subgraphs. Section \ref{sec:peeling_of_a_set} introduces the definition of a peeling of a set; this will be used to reformulate the second stage of Find Erasures and Decode in a way that is suitable for quantum expander codes. Sections \ref{sec:css_codes} and \ref{sec:quantum_expander_codes} state the definitions of Calderbank-Shor-Steane codes and quantum expander codes respectively. 

\subsection{Notation}\label{sec:notation}

$\F$ denotes the field with two elements. $\F^{B\times A}$ denotes the space of matrices with entries in $\F$ whose rows are indexed by finite set $B$, and whose columns are indexed by finite set $A$. 

The row support of $H\in\F^{B\times A}$ is denoted by $\rs(H)\subset B$ and the column support of $H$ is denoted by $\cs(H)\subset A$. 

For $H \in \F^{B\times A}$, $T \subset B$, $S \subset A$, let $\subm{H}{T}{S} \in \F^{T \times S}$ denote the submatrix of $H$ with rows indexed by $T$ and columns indexed by $S$. Let $\esubm{H}{T}{S}$ denote $\subm{H}{T}{S}$ embedded into $\F^{B \times A}$ by padding it with zero rows and columns. If $|T|=|S|$ and $\subm{H}{T}{S}$ is invertible, the inverse is denoted by $\submi{H}{S}{T}\in\F^{S \times T}$. Let $\esubmi{H}{S}{T}$ denote $\submi{H}{S}{T}$ embedded into $\F^{A\times B}$ by padding it with zero rows and columns. Thus, 
\begin{align}
H &= \twobytwo{\subm{H}{T}{S}}{\subm{H}{T}{A\backslash S}}{\subm{H}{B \backslash T}{S}}{\subm{H}{B\backslash T}{A \backslash S}} \\
\esubm{H}{T}{S} &= \twobytwo{\subm{H}{T}{S}}{\subm{0}{T}{A\backslash S}}{\subm{0}{B\backslash T}{S}}{\subm{0}{B\backslash T}{A \backslash S}} \\
\esubmi{H}{S}{T} &= \twobytwo{\submi{H}{S}{T}}{\subm{0}{S}{B\backslash T}}{\subm{0}{A\backslash S}{T}}{\subm{0}{A\backslash S}{B \backslash T}}
\end{align}
From here follow the useful expressions:
\begin{align}
H\esubmi{H}{S}{T} &= \twobytwo{\subm{I}{T}{T}}{\subm{0}{T}{B\backslash T}}{\subm{H}{B\backslash T}{S}\submi{H}{S}{T}}{\subm{0}{B\backslash T}{B\backslash T}} \label{eq:hhi} \\
\esubmi{H}{S}{T} H &= \twobytwo{\subm{I}{S}{S}}{\submi{H}{S}{T}\subm{H}{T}{A\backslash S}}{\subm{0}{A\backslash S}{S}}{\subm{0}{A\backslash S}{A\backslash S}} \label{eq:hih}
\end{align}

The transpose of $H$ is denoted $H^\top$. $\submt{H}{S}{T}$ denotes a submatrix of $H^\top$, $\submti{H}{T}{S}$ denotes its inverse (if it exists), and $\esubmt{H}{S}{T}$ and $\esubmti{H}{T}{S}$ denote the respective embeddings into $\F^{A \times B}$ and $\F^{B \times A}$. 

\subsection{Lossless expanders}\label{sec:lossless_expanders}

In a lossless expander, small sets have nearly as many neighbors as they have edges. 

\begin{definition}\label{def:lossless_expander}
Let $G=(A \cup B, E)$ be a bipartite graph with $A,B$ the two sets of vertices and $E \subset A \times B$ the set of edges. Let $\gamma_A,\gamma_B,\delta_A,\delta_B \in [0,1]$. $G$ is a $(\gamma_A,\gamma_B,\delta_A,\delta_B)$ two-sided lossless expander if 
\begin{align}
S \subset A, 0<|\E(S)| \leq \gamma_A |E| &\implies |\B(S)| > (1-\delta_A) |\E(S)| \label{eq:lossless_expander_a}\\
T \subset B, 0<|\E(T)| \leq \gamma_B |E| &\implies |\A(T)| > (1-\delta_B) |\E(T)| \label{eq:lossless_expander_b}
\end{align}
where $\B(S)$ is the set of vertices in $B$ that are connected to $S$, $\A(T)$ is the set of vertices in $A$ that are connected to $T$, and $\E(S),\E(T)$ are the sets of edges leaving $S,T$ respectively. 
\end{definition}

As in \cite{burshtein2001expander}, the antecedent in \eqref{eq:lossless_expander_a} and \eqref{eq:lossless_expander_b} is a condition on the number of edges. This formulation is suitable for both regular and irregular graphs. 

A random graph from a suitable distribution has $|E|=O(|A|+|B|)$ and is an expander with high probability \cite[Lemma 1]{burshtein2001expander}. 

\subsection{The graph neighborhoods $\Ab$ and $\Bb$}\label{sec:graph_neighborhoods}

It will be convenient to use the following:

\begin{definition}\label{def:beta_neighborhood}
Let $G=(A \cup B, E)$ be a bipartite graph. For $\beta \in [0,1)$, $S \subset A$, let $\Bb(S)$ denote the set of those vertices in $B$  for which more than a $\beta$ fraction of edges go to $S$, i.e. 
\begin{equation}\label{eq:b_beta}
\Bb(S)=\{b\in B:|\A(b)\cap S|> \beta |\E(b)|\}
\end{equation}
and, similarly, for $T \subset B$, let $\Ab(T)$ denote the set of those vertices in $A$ for which more than a $\beta$ fraction of edges go to $T$, i.e.
\begin{equation}\label{eq:a_beta}
\Ab(T)=\{a \in A : |\B(a)\cap T| > \beta |\E(a)|\}
\end{equation}
$\A^0=\A$ and $\B^0=\B$ are the usual graph neighborhoods. 
\end{definition}

$\Ab$ and $\Bb$ are monotone:

\begin{lemma}\label{lemma:monotonicity_of_neighborhood}
For all $\beta \in [0,1)$, if $S \subset S' \subset A$ then $\Bb(S) \subset \Bb(S')$ and if $T \subset T' \subset B$ then $\Ab(T) \subset \Ab(T')$. 
\end{lemma}

\begin{proof}
Take any $b \in \Bb(S)$. Then, $|\A(b)\cap S'| \geq |\A(b)\cap S| > \beta |\E(b)|$, so $b \in \Bb(S')$. The proof for $\Ab$ is similar. 
\end{proof}

\subsection{Sequences of subgraphs}\label{sec:sequences_of_subgraphs}

The Find Erasures and Decode algorithm for classical expander codes has two stages: the first stage finds a small superset of the error positions, and the second stage uses this additional information to correct the errors. It will be convenient to reformulate the first stage as follows:

\begin{definition}\label{def:decoder_sequence}
Let $G=(A \cup B, E)$ be a bipartite graph. For $\beta \in [0,1)$, $S \subset A$, and for $i=0,1,\dots$, let 
\begin{align}
\Gc_i(S)=(\Ac_i(S)\cup\Bc_i(S),E\cap(\Ac_i(S)\times\Bc_i(S)))
\end{align} 
denote the sequence of subgraphs given by
\begin{align}
\Bc_0(S)&=\emptyset  \\
\Ac_0(S) &= S \\
\Bc_{i+1}(S) &= \Bb(\Ac_i(S))  \\
\Ac_{i+1}(S) &=\Ac_i(S) \cup \A(\Bc_{i+1}(S)) 
\end{align}
and let 
\begin{align}
\Ac (S) &= \cup_i \Ac_i(S) \label{eq:mathcal_a_beta_of_s} \\
\Bc (S) &= \cup_i \Bc_i(S) \label{eq:mathcal_b_beta_of_s} \\
\Gc (S) &= (\Ac(S)\cup\Bc(S), E \cap(\Ac(S)\times\Bc(S)))
\end{align}
Similarly, for $T \subset B$, for $i=0,1,\dots$, let 
\begin{align}
\Gc_i(T)=(\Ac_i(T)\cup\Bc_i(T),E\cap(\Ac_i(T)\times\Bc_i(T)))
\end{align} 
denote the sequence of subgraphs given by
\begin{align}
\Ac_0(T) &=\emptyset  \\
\Bc_0(T)&=T  \\
\Ac_{i+1}(T) &=\Ab(\Bc_i(T))  \\
\Bc_{i+1}(T) &=\Bc_i(T)\cup \B(\Ac_{i+1}(T)) 
\end{align}
and let 
\begin{align}
\Ac(T) &= \cup_i \Ac_i(T) \label{eq:mathcal_a_beta_of_t} \\
\Bc(T)&=\cup_i \Bc_i(T) \label{eq:mathcal_b_beta_of_t} \\
\Gc(T)&=(\Ac(T)\cup\Bc(T),E\cap(\Ac(T)\times\Bc(T)))
\end{align}
\end{definition}

The sequences $\{\Gc_i(S)\}_{i=0}^\infty$ and $\{\Gc_i(T)\}_{i=0}^\infty$ are increasing and bounded above, so they become constant after a finite number of terms. Indeed, $\Ac_i(S) \subset \Ac_{i+1}(S)$ by construction, and consequently also $\Bc_{i+1}(S) \subset \Bc_{i+2}(S)$. Similarly, $\Bc_i(T) \subset \Bc_{i+1}(T)$ by construction, and consequently also $\Ac_{i+1}(T) \subset \Ac_{i+2}(T)$. This argument relies on the monotonicity of $\Ab$ and $\Bb$ (Lemma \ref{lemma:monotonicity_of_neighborhood}).

\subsection{Peeling of a set}\label{sec:peeling_of_a_set}

The following definition is inspired by the peeling decoder \cite[Algorithm 1]{luby2001efficient}, \cite[Section 3.19]{richardson2008modern}, \cite[Section 4.2]{viderman2013lineartime}. 

\begin{definition}\label{def:peeling_of_a_set}
Let $G=(A \cup B, E)$ be a bipartite graph. 

A peeling of a set $S \subset A$ is a matching of $S$ consisting of edges $(a_i,b_i)$, $i=1,\dots,|S|$, with the additional property: if $i>j$, then $a_i$ and $b_j$ are not connected. The set $\{b_1,\dots,b_{|S|}\}$ will be denoted $\Bp(S)$. 

Similarly, a peeling of a set $T \subset B$ is a matching of $T$ consisting of edges $(a_i,b_i)$, $i=1,\dots,|T|$, with the additional property: if $i<j$, then $a_i$ and $b_j$ are not connected. The set $\{a_1,\dots,a_{|T|}\}$ will be denoted $\Ap(T)$.
\end{definition}

There may be several subsets of $\B(S)$ which can appear in a peeling of $S$. However, once an algorithm for computing a peeling is fixed, the subset $\Bp(S) \subset \B(S)$ is uniquely determined. Similarly, $\Ap(T) \subset \A(T)$ is uniquely determined by the algorithm used to compute it. 

\subsection{Calderbank-Shor-Steane codes}\label{sec:css_codes}

A quantum CSS code \cite{calderbank1996good}, \cite{steane1996multiple} with set of qubits $Q$ and two sets of checks $C_x$, $C_z$ is specified by matrices $H_x \in \F^{C_x \times Q}$ and $H_z \in \F^{C_z \times Q}$ such that $H_x H_z^\top = 0$. The code has $|Q|$ physical qubits and $|Q|-rank(H_x)-rank(H_z)$ logical qubits. 

The decoding problem for a quantum CSS code is the following: given the syndromes $W=H_z X$, $U=H_x Z$ of errors $X,Z \in \F^Q$, find corrections $\hat{X},\hat{Z}$ with the property: $\hat{X}+X\in Im(H_x^\top)$ and $\hat{Z}+Z\in Im(H_z^\top)$. 

\subsection{Quantum expander codes}\label{sec:quantum_expander_codes}

Quantum expander codes \cite{leverrier2015quantum} are CSS codes obtained from the hypergraph product \cite{tillich2014quantum} of a two-sided lossless expander with itself. 

\begin{definition}\label{def:hypergraph_product}
The hypergraph product of $G=(A \cup B, E)$ with itself has four sets of vertices $A \times A$, $A \times B$, $B \times B$, $B \times A$, and has the following edges: for all $a' \in A$, $b' \in B$, $(a,b) \in E$ 
\begin{enumerate}
\item $(a',a) \in A \times A$ and $(a',b) \in A \times B$ are connected.
\item $(a,b') \in A \times B$ and $(b,b') \in B \times B$ are connected.
\item $(b',b) \in B \times B$ and $ (b',a) \in B \times A$ are connected.
\item $(b,a') \in B \times A$ and $(a,a') \in A \times A$ are connected. 
\end{enumerate}

The quantum CSS code associated to $G$ has qubits indexed by $A\times A \cup B \times B$, X stabilizer generators indexed by $A \times B$ and Z stabilizer generators indexed by $B \times A$. It has blocklength $|A|^2 + |B|^2$ and encodes $(|A|-rank(H))^2+(|B|-rank(H))^2$ logical qubits. Moreover, if $H \in \F^{B \times A}$ is the adjacency matrix of $G$, then
\begin{enumerate}
\item Any Pauli-X error $(X_A,X_B)$ with $X_A\in \F^{A\times A}$ and $X_B \in \F^{B \times B}$ produces syndrome $HX_A + X_BH \in \F^{B\times A}$. 
\item Any Pauli-X stabilizer is of the form $(VH,HV)$ for some $V \in \F^{A \times B}$. 
\item Any Pauli-Z error $(Z_A,Z_B)$ with $Z_A \in\F^{A \times A}$ and $Z_B \in \F^{B\times B}$ produces syndrome $Z_AH^\top+H^\top Z_B \in \F^{A\times B}$.
\item Any Pauli-Z stabilizer is of the form $(H^\top V,VH^\top)$ for some $V \in \F^{B\times A}$. 
\end{enumerate}
\end{definition}

\section{Find Rows and Columns and Decode}\label{sec:find_rows_and_columns_and_decode}

Section \ref{sec:algorithm} introduces the algorithm and states the Theorem about its decoding radius and complexity. Section \ref{sec:overview_of_the_proof} gives an overview of the proof and Sections \ref{sec:proof_that_execution_is_contained}, \ref{sec:proof_that_output_is_correct} and \ref{sec:proof_of_linear_time} contain the details. 

\subsection{The algorithm, its decoding radius and its complexity}\label{sec:algorithm}

\begin{definition}\label{def:find-row-find-column}
The Find Rows and Columns and Decode algorithm for quantum expander codes is parametrized by a bipartite graph $G=(A \cup B, E)$, by $\beta \in [0,1)$ and by $\ell \in \mathbb{N}$. On input a syndrome $W\in\F^{B\times A}$ of a Pauli-X error, it takes the following steps:
\begin{enumerate}
\item It computes the row support of $W$, denoted by $\rs(W)\subset B$, and the column support of $W$, denoted by $\cs(W)\subset A$. 
\item It computes the sets $\Ac_\ell(\rs(W)) \subset A$ and $\Bc_\ell(\cs(W)) \subset B$ (recall Definition \ref{def:decoder_sequence}). 
\item It computes a peeling of $\Ac_\ell(\rs(W))$ and a peeling of $\Bc_\ell(\cs(W))$ (recall Definition \ref{def:peeling_of_a_set}). 
\item It outputs the corrections
\begin{align}
\hat{X}_A &= \esubmi{H}{\Ac_\ell(\rs(W))}{\Bp(\Ac_\ell(\rs(W)))}W \label{eq:correction_x_a} \\
\hat{X}_B &= (H\hat{X}_A+W)\esubmi{H}{\Ap(\Bc_\ell(\cs(W)))}{\Bc_\ell(\cs(W))} \label{eq:correction_x_b}
\end{align}
where the notation of Section \ref{sec:notation} is used. 
\end{enumerate}
A Pauli-Z error is corrected similarly: on input a syndrome $U \in \F^{A \times B}$ of a Z error, the Find Rows and Columns and Decode algorithm outputs the corrections
\begin{align}
\hat{Z}_A &= U \esubmti{H}{\Bp(\Ac_\ell(\cs(U)))}{\Ac_\ell(\cs(U))} \label{eq:correction_z_a} \\
\hat{Z}_B &= \esubmti{H}{\Bc_\ell(\rs(U))}{\Ap(\Bc_\ell(\rs(U)))}(U+\hat{Z}_AH^\top)\label{eq:correction_z_b}
\end{align}
\end{definition}

The Find Rows and Columns and Decode algorithm can correct all errors up to a constant fraction of the code distance: 

\begin{theorem}\label{thm:decoding_radius_and_complexity}
Let $G=(A \cup B, E)$ be a $(\gamma_A$, $\gamma_B$, $\delta_A$, $\delta_B)$ two-sided lossless expander. Let the maximum left and right degree of $G$ be $\dax$, $\dbx$ respectively. Let $\delta = \max(\delta_A,\delta_B) < 1/3$ and let $\delta < \beta \leq 1-2\delta$. Let
\begin{align}
\ell_A &= \min\left(|A|,\left\lceil\frac{\log|E|}{\log\frac{1-\beta}{2\delta_A}}\right\rceil\right) \\
\ell_B &=\min\left(|B|,\left\lceil\frac{\log|E|}{\log\frac{1-\beta}{2\delta_B}}\right\rceil\right) \\
\ell &=\max(\ell_A,\ell_B)
\end{align}
Let $W$, $U$ be syndromes that are generated respectively by a Pauli-X error $(X_A,X_B)$ and a Pauli-Z error $(Z_A,Z_B)$ of Hamming weight less than or equal to 
\begin{align}
\frac{\beta-\delta}{\beta}\left(\min\left(\gamma_A \frac{|E|}{\dax},\gamma_B \frac{|E|}{\dbx}\right)-1+\min\left(\frac{1}{d_A},\frac{1}{d_B}\right)\right)
\end{align}
Then, the Find Rows and Columns and Decode algorithm outputs a correction that differs from the error by a stabilizer. The algorithm runs in $O(|E|(|A|+|B|))$ time in the uniform cost model. In addition, it is parallelizable to a circuit with depth $O((\ell+\frac{\log |E|}{\log\frac{1}{2\delta}})(\log\dax+\log\dbx))$ and with $O(|E|(|A|+|B|)\frac{\log|E|}{\log\frac{1}{2\delta}})$ gates, where $\log$ denotes the binary logarithm. 
\end{theorem}

Some remarks on important special cases of Theorem \ref{thm:decoding_radius_and_complexity} follow. When $|E|=O(|A|+|B|)$, Find Rows and Columns and Decode runs in time linear in the blocklength $|A|^2+|B|^2$. When the sequential algorithm is applied, it makes sense to choose $\beta=1-2\delta$ to correct the largest number of errors. 

When the maximum degrees $\dax$ and $\dbx$ are constant and $\beta=1-2\delta$, the parallel algorithm has depth that scales with the square root of the blocklength. Decreasing $\beta$ slightly so that it is bounded away from $1-2\delta$ results in a circuit of logarithmic depth, but also in a slightly smaller decoding radius. One is reminded of \cite[Remark 13]{sipser1996expander} that in experiments, the sequential Flip Algorithm usually corrects more errors than the parallel one. 

In the case of Find Rows and Columns and Decode, the difference between the sequential and parallel version arises in the interplay of Lemma \ref{lemma:lower_bound_on_decoder_sequence}, which gives bounds on how many terms of the sequences $\{\Gc_i(\cs(W))\}_{i=0}^\infty$ and $\{\Gc_i(\rs(W))\}_{i=0}^\infty$ need to be computed, and Lemma \ref{lemma:computation_of_sequences}, which discusses the complexity of computing them. For the sequential algorithm, it does not matter if $O(\log|E|)$ or $\Theta(|E|)$ terms of the sequences are needed; in both cases, it takes time $O(|E|)$ in the uniform cost model. However, the depth of the parallel algorithm scales with the number of terms, and an example after Lemma \ref{lemma:lower_bound_on_decoder_sequence} illustrates the need to decrease $\beta$ slightly to make sure that a logarithmic number of terms suffice.  

Finally, in the special case when all vertices on the same side have the same degree, 
\begin{align}
\min\left(\gamma_A\frac{|E|}{\dax},\gamma_B\frac{|E|}{\dbx}\right)=\min(\gamma_A|A|,\gamma_B|B|)
\end{align} 
which is the usual expression for the expansion radius of regular bipartite graphs. 

\subsection{Overview of the proof of Theorem \ref{thm:decoding_radius_and_complexity}}\label{sec:overview_of_the_proof}

The argument is presented for the correction of an X error; the case of a Z error is similar. 

First, the execution of the algorithm is contained in a small region and this guarantees that a peeling for $\Ac_\ell(\rs(W))$ and for $\Bc_\ell(\cs(W))$ can be found:

\begin{proposition}\label{prop:execution_is_contained}
There are sets $C\subset A$, $\Acb(C) \subset A$, $D \subset B$, $\Bcb(D)\subset B$ such that
\begin{align}
\Ac(\rs(W)) &\subset \Acb(C) \label{eq:container_for_mathcal_a_beta_of_t} \\
\Bc(\cs(W)) &\subset \Bcb(D) \label{eq:container_for_mathcal_b_beta_of_s} \\
|\E(\Acb(C))| & \leq \gamma_A |E| \label{eq:size_of_bar_mathcal_a_of_c} \\
|\E(\Bcb(D))| & \leq \gamma_B |E| \label{eq:size_of_bar_mathcal_b_of_d}
\end{align}
\end{proposition}

\begin{proof}
The sets $C,D$ are chosen to contain the error $(X_A,X_B)$ in the sense of Definition \ref{def:cd_contained}. Lemma \ref{lemma:size_of_container_for_an_error} implies that $C,D$ can be chosen so that each of $|C|$, $|D|$ is at most the Hamming weight $|X_A|+|X_B|$. Lemma \ref{lemma:container_for_an_error_is_container_for_the_syndrome} implies that $\rs(W) \subset \B(C)$ and $\cs(W) \subset \A(D)$. 

The sets $\Acb(C)$ and $\Bcb(D)$ are constructed in Definition \ref{def:upper_bound_sequence}. Lemma \ref{lemma:upper_bound_sequence_and_find_erasures_sequence} implies that $\Ac(\rs(W)) \subset \Acb(C)$ and $\Bc(\cs(W)) \subset \Bcb(D)$. Lemma \ref{lemma:upper_bound_on_bar_a_of_c_and_bar_b_of_d} implies that $|\E(\Acb(C))|\leq\gamma_A|E|$ and $|\E(\Bcb(D))|\leq\gamma_B|E|$.

The details of the proof of Proposition \ref{prop:execution_is_contained} are in Section \ref{sec:proof_that_execution_is_contained}.
\end{proof}

Next, the output of the algorithm is correct:

\begin{proposition}\label{prop:output_is_correct}
The correction $(\hat{X}_A,\hat{X}_B)$ differs from the error $(X_A,X_B)$ by a Pauli-X stabilizer. 
\end{proposition}

\begin{proof}
Definition \ref{def:improved_lower_bound_on_decoder_sequence} introduces sets $\Acu_C(\rs(W))$ and $\Bcu_D(\cs(W))$ and Lemma \ref{lemma:improved_lower_bound_on_decoder_sequence} shows that
\begin{align}
\Acu_C(\rs(W)) &\subset \Ac_\ell(\rs(W)) \\
\Bcu_D(\cs(W)) &\subset \Bc_\ell(\cs(W))
\end{align}
At the same time, Lemma \ref{lemma:row_support_included_in_lower_bound_set} shows that there is a Pauli-X error $(X_A',X_B')$ that is stabilizer equivalent to $(X_A,X_B)$, is $(C,D)$ contained, and has the additional property that the row support of $X_A'$ is a subset of $\Acu_C(\rs(W))$. 

After the correction $\hat{X}_A$, the remaining error is $(X_A'+\hat{X}_A,X_B')$ and the remaining syndrome is $W+H\hat{X}_A$. The row support of $X_A'+\hat{X}_A$ is a subset of $\Ac_\ell(\rs(W))$ and the rows of the remaining syndrome $W+H\hat{X}_A$ indexed by $\Bp(\Ac_\ell(\rs(W)))$ are zero (Lemma \ref{lemma:support_of_remaining_syndrome}), so $(X_A'+\hat{X}_A,X_B')$ is stabilizer equivalent to $(0,X_B'')$ for some $X_B''$ with column support a subset of $D$ (Lemma \ref{lemma:stabilizer_equivalent_error}). Moreover, the remaining syndrome $X_B''H$ has column support a subset of $\cs(W)$ (Lemma \ref{lemma:support_of_remaining_syndrome}), so $\cs(X_B'')$ must be a subset of $\Bcu_D(\cs(W))$ (Lemma \ref{lemma:about_column_support}). Therefore, $\hat{X}_B=X_B''$ (Lemma \ref{lemma:b_side_correction_equals_remaining_error}). 

The details of the proof of Proposition \ref{prop:output_is_correct} are in Section \ref{sec:proof_that_output_is_correct}. 
\end{proof}

Finally, the algorithm runs in linear time and can be parallelized to logarithmic depth:

\begin{proposition}\label{prop:linear_time}
The algorithm runs in $O(|E|(|A|+|B|))$ time in the uniform cost model and can be parallelized to a circuit with depth $O((\ell+\frac{\log|E|}{\log\frac{1}{2\delta}})(\log\dax+\log\dbx))$ and with $O(|E|(|A|+|B|)\frac{\log|E|}{\log\frac{1}{2\delta}})$ gates. 
\end{proposition}

\begin{proof}
The row and column support of the syndrome can be computed in time $O(|A||B|)$ in the uniform cost model or by a circuit of depth $O(\log|A|+\log|B|)$ with $O(|A||B|)$ gates (Lemma \ref{lemma:computation_of_the_row_and_column_support_of_the_syndrome}). The sets $\Ac_\ell(\rs(W))$ and $\Bc_\ell(\cs(W))$ can be computed in time $O(|E|)$ or by a circuit with depth $O(\ell(\log\dax+\log\dbx))$ and with $O(\ell|E|)$ gates (Lemma \ref{lemma:computation_of_sequences}).The peelings of the sets $\Ac_\ell(\rs(W))$ and $\Bc_\ell(\cs(W))$ can be computed in time $O(|E|)$ or by a circuit of depth $O(\frac{\log|E|}{\log\frac{1}{2\delta}}(\log\dax+\log\dbx))$ and with $O(|E|\frac{\log|E|}{\log\frac{1}{2\delta}}(\dax+\dbx))$ gates (Lemma \ref{lemma:computation_of_peelings}). The corrections $\hat{X}_A$ and $\hat{X}_B$ can be computed in time $O(|E|(|A|+|B|)$ or by a circuit with depth $O(\frac{\log|E|}{\log\frac{1}{2\delta}}(\log\dax+\log\dbx))$ and with $O\left(|E|(|A|+|B|)\frac{\log|E|}{\log\frac{1}{2\delta}}\right)$ gates (Lemma \ref{lemma:applying_the_inverse_of_a_submatrix}).

The details of the proof of Proposition \ref{prop:linear_time} are in Section \ref{sec:proof_of_linear_time}.
\end{proof}

\subsection{Details of the proof of Proposition \ref{prop:execution_is_contained}: the execution of the algorithm is contained in a small region}\label{sec:proof_that_execution_is_contained}

\subsubsection{The error is contained in a small region}\label{sec:the_error_is_contained}

The sets $C,D$ in Proposition \ref{prop:execution_is_contained} are chosen to contain the error in the following sense: 

\begin{definition}\label{def:cd_contained}
For $X_A \in \F^{A \times A}$, $X_B \in \F^{B \times B}$, $C\subset A$, $D \subset B$, say that $(X_A,X_B)$ is $(C,D)$ contained if the row support of $X_A$ is a subset of $C$, the column support of $X_A$ is a subset of $\A(D)$, the row support of $X_B$ is a subset of $\B(C)$ and the column support of $X_B$ is a subset of $D$. 
\end{definition}

The size of the container for a small error can be small: 

\begin{lemma}\label{lemma:size_of_container_for_an_error}
For all $(X_A,X_B)$ there exist $(C,D)$ such that $(X_A,X_B)$ is $(C,D)$ contained and each of $|C|$, $|D|$ is at most $|X_A|+|X_B|$. 
\end{lemma}

\begin{proof}
To build $C$, first add the row support of $X_A$; then, for each non-zero row of $X_B$, add one of its neighbors. The argument for $D$ is similar. 
\end{proof}

The container for an error is also a container for the syndrome:

\begin{lemma}\label{lemma:container_for_an_error_is_container_for_the_syndrome}
If $(X_A,X_B)$ is $(C,D)$ contained, then the row support of $HX_A+X_BH$ is a subset of $\B(C)$, and the column support of $HX_A+X_BH$ is a subset of $\A(D)$. 
\end{lemma}

\begin{proof}
Take $b\in B$ such that row $b$ of $HX_A+X_BH$ is not zero. Then either row $b$ of $X_B$ is not zero, or $b$ is a neighbor to $a\in A$ such that row $a$ of $X_A$ is not zero. In both cases, $b\in \B(C)$. The argument for columns is similar. 
\end{proof}

\subsubsection{Upper bounds on the sequences $\{\Gc_i(T)\}_{i=0}^\infty$ and $\{\Gc_i(S)\}_{i=0}^\infty$}\label{sec:upper_bounds_on_decoder_sequence}

The sets $\Acb(C)$ and $\Bcb(D)$ in Proposition \ref{prop:execution_is_contained} are constructed using the following sequences of subgraphs: 

\begin{definition}\label{def:upper_bound_sequence}
For $\beta \in [0,1)$, $C \subset A$, and $i=0,1,\dots$, let
\begin{align}
\Gcb_i(C)=(\Acb_i(C)\cup\Bcb_i(C),E\cap(\Acb_i(C)\times\Bcb_i(C)))
\end{align}
be the sequence of subgraphs given by 
\begin{align}
\Acb_0(C) &= C \\
\Bcb_0(C) &= \B(C) \\
\Acb_{i+1}(C) &= \Ab(\Bcb_i(C))\\
\Bcb_{i+1}(C) &= \B(\Acb_{i+1}(C)) 
\end{align}
and let
\begin{align}
\Acb(C) &= \cup_i \Acb_i(C) \label{eq:bar_mathcal_a_beta_of_c}\\
\Bcb(C) &= \cup_i \Bcb_i(C) \label{eq:bar_mathcal_b_beta_of_c}\\
\Gcb(C) &=(\Acb(C)\cup\Bcb(C),E\cap(\Acb(C)\times\Bcb(C)))
\end{align}

Similarly, for $D \subset B$, $i=0,1,\dots$, let
\begin{align}
\Gcb_i(D)=(\Acb_i(D)\cup\Bcb_i(D),E\cap(\Acb_i(D)\cap\Bcb_i(D))) 
\end{align}
be the sequence of subgraphs given by
\begin{align}
\Bcb_0(D)&=D\\
\Acb_0(D)&=\A(D)\\
\Bcb_{i+1}(D)&=\Bb(\Acb_i(D))\\
\Acb_{i+1}(D)&=\A(\Bcb_{i+1}(D))
\end{align}
and let
\begin{align}
\Acb(D)&=\cup_i\Acb_i(D)\label{eq:bar_mathcal_a_beta_of_d}\\
\Bcb(D)&=\cup_i\Bcb_i(D)\label{eq:bar_mathcal_b_beta_of_d}\\
\Gcb(D)&=(\Acb(D)\cup\Bcb(D),E\cap(\Acb(D)\times\Bcb(D)))
\end{align}
\end{definition}

The sequences $\{\Gcb_i(C)\}_{i=0}^\infty$ and $\{\Gcb_i(D)\}_{i=0}^\infty$ are increasing and bounded, so they become constant after a finite number of terms. Indeed, all edges from $\Acb_i(C)$ go to $\Bcb_i(C)$, so $\Acb_i(C)$ is a subset of $\Ab(\Bcb_i(C))=\Acb_{i+1}(C)$ and consequently also $\Bcb_i(C)$ is a subset of $\Bcb_{i+1}(C)$. The argument for $\Gcb_i(D)$ is similar. 

The sequences $\Gcb_i(C)$ and $\Gcb_i(D)$ can be used as upper bounds on the sequences $\Gc_i(T)$ and $\Gc_i(S)$; moreover, in some cases, the bounds are tight:

\begin{lemma}\label{lemma:upper_bound_sequence_and_find_erasures_sequence}
For a bipartite graph $G=(A \cup B,E)$, $C,S \subset A$, $D, T \subset B$ and $\beta \in [0,1)$:
\begin{enumerate}
\item If $T\subset\B(C)$ then for all $i$, $\Ac_i(T)\subset\Acb_i(C)$ and $\Bc_i(T)\subset\Bcb_i(C)$. 
\item If there is a $k$ such that $C\subset \Ac_k(T)$ then for all $i$, $\Acb_i(C)\subset \Ac_{i+k}(T)$ and $\Bcb_{i}(C) \subset \Bc_{i+k}(T)$. 
\item If $S\subset\A(D)$ then for all $i$, $\Ac_i(S)\subset\Acb_i(D)$ and $\Bc_i(S)\subset\Bcb_i(D)$. 
\item If there is a $k$ such that $D\subset \Bc_k(S)$ then for all $i$, $\Acb_i(D) \subset \Ac_{i+k}(S)$ and $\Bcb_i(D) \subset \Bc_{i+k}(S)$. 
\end{enumerate}
\end{lemma}

\begin{proof}
The proof of part 1 is by induction on $i$. The base case $i=0$ holds: 
\begin{align}
\Ac_0(T)=\emptyset & \subset C = \Acb_0(C) \\
\Bc_0(T)=T &\subset \B(C) = \Bcb_0(C)
\end{align}
If, for some $i \geq 0$, $\Ac_i(T)\subset \Acb_i(C)$ and $\Bc_i(T) \subset \Bcb_i(C)$, then
\begin{align}
\Ac_{i+1}(T) = \Ab(\Bc_i(T)) &\subset \Ab(\Bcb_i(C)) = \Acb_{i+1}(C)
\end{align}
and
\begin{align}
\Bc_{i+1}(T) &= \Bc_i(T) \cup \B(\Ac_{i+1}(T)) \\
&\subset \Bcb_i(C) \cup \B(\Acb_{i+1}(C)) = \Bcb_{i+1}(C)
\end{align}
which completes the inductive step. 

The proof of part 2 is also by induction on $i$. The base case $i=0$ holds:
\begin{align}
\Acb_0(C) =C &\subset \Ac_k(T) \\
\Bcb_0(C)=\B(C) \subset \B(\Ac_k(T))&\subset \Bc_k(T)
\end{align}
If, for some $i \geq 0$, $\Acb_i(C) \subset \Ac_{i+k}(T)$ and $\Bcb_i(C) \subset \Bc_{i+k}(T)$, then
\begin{align}
\Acb_{i+1}(C) = \Ab(\Bcb_i(C)) &\subset \Ab(\Bc_{i+k}(T)) = \Ac_{i+1+k}(T) \\
\Bcb_{i+1}(C) = \B(\Acb_{i+1}(C)) &\subset \B(\Ac_{i+1+k}(T)) \subset \Bc_{i+1+k}(T)
\end{align}
which completes the inductive step.

The proofs of parts 3 and 4 are similar. 
\end{proof}

In the sequence $\Gcb_i(C)$, the growth of the number of vertices on the right is at most a $(1-\beta)$ fraction of the growth of the number of edges, and, similarly, for $\Gcb_i(D)$, the growth of the number of vertices on the left is at most a $(1-\beta)$ fraction of the growth of the number of edges: 

\begin{lemma}\label{lemma:growth_of_edges_and_growth_of_vertices_for_upper_bound_sequence}
For a bipartite graph $G=(A \cup B,E)$, $C \subset A$, $D \subset B$ and $\beta\in[0,1)$:
\begin{enumerate}
\item For every $i$ and every set $C'$ with $\Acb_i(C) \subset C' \subset \Acb_{i+1}(C)$, 
\begin{align}\label{eq:growth_of_edges_and_growth_of_vertices_for_c}
|\B(C')|-|\B(C)| < (1-\beta) (|\E(C')|-|\E(C)|)
\end{align}
\item For every $i$ and every set $D'$ with $\Bcb_i(D) \subset D' \subset \Bcb_{i+1}(D)$, 
\begin{align}\label{eq:growth_of_edges_and_growth_of_vertices_for_d}
|\A(D')|-|\A(D)| < (1-\beta)(|\E(D')|-|\E(D)|)
\end{align}
\end{enumerate}
\end{lemma}

\begin{proof}
The proof of part 1 is by induction on $i$. The base case $i=0$ holds: if $C \subset C' \subset \Acb_1(C)$ then
\begin{align}
&|\B(C')|-|\B(C)| \\
&\leq \sum_{c \in C' \backslash C} |\B(c)\backslash\B(C)| \\
& < \sum_{c \in C' \backslash C} (1-\beta) |\E(c)| \\
& = (1-\beta) (|\E(C')|-|\E(C)|)
\end{align}
Suppose \eqref{eq:growth_of_edges_and_growth_of_vertices_for_c} holds for any $C'$ with $\Acb_i(C) \subset C' \subset \Acb_{i+1}(C)$; in particular,  
\begin{align}
|\B(\Acb_{i+1})(C)| -|\B(C)| < (1-\beta)(|\E(\Acb_{i+1}(C))|-|\E(C)|)
\end{align}
Now, take any $C'$ with $\Acb_{i+1}(C) \subset C' \subset \Acb_{i+2}(C)$. By the same reasoning as in the base case,
\begin{align}
|\B(C')|-|\B(\Acb_{i+1}(C))|<(1-\beta)(|\E(C')|-|\E(\Acb_{i+1}(C))|)
\end{align}
Adding this and the previous inequality completes the inductive step. 

The proof of part 2 is similar. 
\end{proof}

Combined with the expansion of the graph and the upper bound on the size of $C$ and $D$, Lemma \ref{lemma:growth_of_edges_and_growth_of_vertices_for_upper_bound_sequence} implies an upper bound on the size of $\E(\Acb(C))$ and $\E(\Bcb(D))$: 

\begin{lemma}\label{lemma:upper_bound_on_bar_a_of_c_and_bar_b_of_d}
Let $G=(A \cup B, E)$ be a bipartite $(\gamma_A,\gamma_B,\delta_A,\delta_B)$ two-sided lossless expander with maximum left and right degrees $\dax,\dbx$. Let $\delta=\max(\delta_A,\delta_B)$ and $\beta > \delta$. 
\begin{enumerate} 
\item If $C\subset A$ and $|\E(C)|\leq \frac{\beta-\delta}{\beta}(\gamma_A|E|-\dax+1)$, then $|\E(\Acb(C))|< \gamma_A|E|$. 
\item If $D\subset B$ and $|\E(D)| \leq \frac{\beta-\delta}{\beta}(\gamma_B|E|-\dbx+1)$, then $|\E(\Bcb(D))| < \gamma_B|E|$. 
\end{enumerate}
\end{lemma}

\begin{proof}
Suppose for a contradiction that for some $i$, 
\begin{align}
|\E(\Acb_i(C))| < \gamma_A|E|\leq|\E(\Acb_{i+1}(C))|
\end{align}
Take $C'$ such that 
\begin{align}
\Acb_i(C) \subset C' \subset \Acb_{i+1}(C)\\
\gamma_A|E| - \dax + 1 \leq |\E(C')|\leq\gamma_A|E|
\end{align}
Lemma \ref{lemma:growth_of_edges_and_growth_of_vertices_for_upper_bound_sequence} and the properties of the graph imply that
\begin{align}
|\B(C')|-|\B(C)| &< (1-\beta)(|\E(C')|-|\E(C)|) \\
|\B(C')| &> (1-\delta)|\E(C')| \\
|\B(C)| &\leq |\E(C)| 
\end{align}
This system of inequlities implies that $(\beta-\delta)|\E(C')|<\beta|\E(C)|$, a contradiction. 

The proof of part 2 is similar. 
\end{proof}

\subsection{Details of the proof of Proposition \ref{prop:output_is_correct}: the output is correct}\label{sec:proof_that_output_is_correct}

\subsubsection{Unique neighbor expanders}\label{sec:unique_neighbor_expanders}

The next steps of the argument use unique neighbor expanders. In such a graph, small sets have many unique neighbors:

\begin{definition}
Let $G=(A \cup B, E)$ be a bipartite graph. For $S \subset A$, let $\Bu(S)$ denote the set of vertices in $B$ that are connected to $S$ by exactly one edge. Similarly, for $T \subset B$, let $\Au(T)$ denote the set of vertices in $A$ that are connected to $T$ by exactly one edge. Let $\gamma_A,\gamma_B,\delta_A,\delta_B \in [0,1]$. $G$ is a $(\gamma_A,\gamma_B,\delta_A,\delta_B)$ two-sided unique neighbor expander if 
\begin{align}
S \subset A, 0<|\E(S)| \leq \gamma_A |E| &\implies |\Bu(S)| > (1-\delta_A) |\E(S)| \label{eq:unique_neighbor_expander_a}\\
T \subset B, 0<|\E(T)| \leq \gamma_B |E| &\implies |\Au(T)| > (1-\delta_B) |\E(T)| \label{eq:unique_neighbor_expander_b}
\end{align}
\end{definition}

Lossless expansion implies unique neighbor expansion: 

\begin{lemma}\label{lemma:lossless_implies_unique_neighbor}
If $G$ is a $(\gamma_A,\gamma_B,\delta_A,\delta_B)$ two-sided lossless expander with $\delta_A,\delta_B \leq 1/2$, then $G$ is also a $(\gamma_A,\gamma_B,2\delta_A,2\delta_B)$ two-sided unique neighbor expander. 
\end{lemma}

\begin{proof}
Take any $S\subset A$ with $0<|\E(S)|\leq\gamma_A|E|$. Count the edges in $\E(S)$ according to the vertices in $\B(S)$, then rearrange and apply lossless expansion:
\begin{align}
|\E(S)| &\geq |\Bu(S)| + 2|\B(S)\backslash\Bu(S)| \\
& = 2 |\B(S)| - |\Bu(S)| \\
& > 2(1-\delta_A) |\E(S)| - |\Bu(S)|
\end{align}
The first part follows. The second part is similar. 
\end{proof}

Moreover, in a unique neighbor expander, the following holds:

\begin{lemma}\label{lemma:less_than_beta_fraction_unique_neighbors}
Let $G=(A \cup B, E)$ be a $(\gamma_A,\gamma_B,\delta_A,\delta_B)$ two-sided unique neighbor expander. Let $\beta \in [0,1)$. Let $S \subset A$ and $T \subset B$ be non-empty and such that $|\E(S)| \leq \gamma_A |E|$ and $|\E(T)| \leq \gamma_B |E|$. Let
\begin{align}
S' &= \{s \in S : |\B(s)\cap\Bu(S)|\leq\beta |\E(s)|\} \\
T' &= \{t \in T : |\A(t)\cap\Au(T)|\leq\beta |\E(t)|\}
\end{align}
Then,
\begin{align}
|\E(S')| &< \frac{\delta_A}{1-\beta} |\E(S)| \\
|\E(T')| &< \frac{\delta_B}{1-\beta} |\E(T)|
\end{align}
\end{lemma}

\begin{proof}
The first part follows from
\begin{align}
&(1 - \delta_A) |\E(S)| \\
&< |\Bu(S)| \\
& \leq \beta |\E(S')| + |\E(S\backslash S')|
\end{align}
by rearrangement. The second part is similar. 
\end{proof}

\subsubsection{Lower bounds on the sequences $\{\Gc_i(S)\}_{i=0}^\infty$ and $\{\Gc_i(T)\}_{i=0}^\infty$}\label{sec:lower_bound_on_decoder_sequence}

The unique neighbors of a union and a difference satisfy:

\begin{lemma}\label{lemma:unique_neighbors_of_union_and_difference}
Let $G=(A \cup B, E)$ be a bipartite graph, and let $S,S' \subset A$, $T,T' \subset B$. Then
\begin{align}
\Bu(S \cup S') &\subset \Bu(S)\cup\Bu(S') \\
\Au(T \cup T') &\subset \Au(T)\cup\Au(T') \\
\Bu(S \backslash S') &\subset \Bu(S) \cup \B(S\cap S') \\
\Au(T \backslash T') &\subset \Au(T) \cup \A(T\cap T')
\end{align}
\end{lemma}

\begin{proof}
Let the unique edge from $b \in \Bu(S \cup S')$ to $S \cup S'$ be $(a,b)$. If $a \in S$, then $b \in \Bu(S)$, and if $a \in S'$ then $b \in \Bu(S')$. 

Next, take $b \in \Bu(S\backslash S')$. If there is only one edge from $b$ to $S$, then $b \in \B(S)$. If there is more than one edge from $b$ to $S$, then $b \in \B(S \cap S')$. 

The other two parts are similar. 
\end{proof}

Therefore, the sequences $\{\Gc_i(T)\}_{i=0}^\infty $ and $\{\Gc_i(S)\}_{i=0}^\infty$ have the following invariants: 

\begin{lemma}\label{lemma:unique_neighbor_inclusion_invariant}
Let $G=(A \cup B, E)$ be a bipartite graph, let $S \subset A$, $T \subset B$ and $\beta \in [0,1)$. If $\Bu(S) \subset T$, then for all $i$, $\Bu(S\backslash \Ac_i(T))\subset\Bc_i(T)$. Similarly, if $\Au(T) \subset S$, then for all $i$, $\Au(T\backslash \Bc_i(S)) \subset \Ac_i(S)$.
\end{lemma}

\begin{proof}
By Lemma \ref{lemma:unique_neighbors_of_union_and_difference}, 
\begin{align}
&\Bu(S\backslash \Ac_i(T)) \\
&\subset \Bu(S) \cup \B(S \cap \Ac_i(T)) \\
& \subset T \cup \B(\Ac_i(T)) \\
&\subset \Bc_i(T)
\end{align}
The other part is similar. 
\end{proof}

Therefore, in a unique neighbor expander, the sequences $\{\Gc_i(T)\}_{i=0}^\infty$ and $\{\Gc_i(S)\}_{i=0}^\infty$ have the following lower bounds: 

\begin{lemma}\label{lemma:lower_bound_on_decoder_sequence}
Let $G=(A \cup B, E)$ be a $(\gamma_A,\gamma_B,\delta_A,\delta_B)$ two-sided unique neighbor expander. Let $\delta=\max(\delta_A,\delta_B)$ and let $\beta \in [0,1-\delta]$. Let
\begin{align}
\ell_A &= \min\left(|A|,\left\lceil\frac{\log|E|}{\log\frac{1-\beta}{\delta_A}}\right\rceil\right) \\
\ell_B &= \min\left(|B|,\left\lceil \frac{\log|E|}{\log\frac{1-\beta}{\delta_B}} \right\rceil\right)
\end{align}
Let $S \subset A$, $T \subset B$. Then,
\begin{enumerate}
\item If $ S' \subset A$, $|\E(S')|\leq \gamma_A|E|$ and $\Bu(S') \subset T$, then $S' \subset \Ac_{\ell_A}(T)$. 
\item If $T' \subset B$, $|\E(T')|\leq \gamma_B |E|$ and $\Au(T') \subset S$, then $T' \subset \Bc_{\ell_B}(S)$. 
\end{enumerate}
\end{lemma}

\begin{proof}
Lemma \ref{lemma:unique_neighbor_inclusion_invariant} implies that for all $i$, $\Bu(S' \backslash \Ac_i(T)) \subset \Bc_i(T)$. While $S' \backslash \Ac_i(T)$ is not empty, Lemma \ref{lemma:less_than_beta_fraction_unique_neighbors} implies that 
\begin{align}
|\E(S'\backslash \Ac_{i+1}(T))| < \frac{\delta_A}{1-\beta} |\E(S'\backslash\Ac_i(T))|
\end{align}
Therefore, $S' \subset \Ac_{\ell_A}(T)$. The proof of the second part is similar. 
\end{proof}

In Lemma \ref{lemma:lower_bound_on_decoder_sequence} there is a distinction between the cases $\beta=1-\delta$ and $\beta < 1-\delta$. This has the following consequence: in order to run Find Rows and Columns and Decode in logarithmic depth, the parameter $\beta$ has to be slightly lower than for the sequential algorithm, and therefore slightly less errors can be corrected. Is this necessary or an artefact of the proof? To answer this, an example is given that suggests it is necessary. Let $S'=\{a_1,\dots,a_n\}$. For each $i$, let vertex $a_i$ of $S'$ be connected to vertices $b_{4i-4}$, $b_{4i-3}$, $b_{4i-2}$, $b_{4i-1}$, $b_{4i}$ in $B$. Each subset of $k$ vertices of $S'$ has $5k$ edges and at least $3k+2$ unique neighbors, so this can be a fragment of a unique neighbor expander with parameter $\delta=2/5$. Now take $T$ to consist of all unique neighbors of $S'$: $b_i$ with $i$ not divisible by 4 as well as $b_0$ and $b_{4n}$. Consider cases. If $\beta=1-\delta$, then strictly more than 3 edges to $\Bc_i(T)$ are needed for a vertex in $S'$ to be added to $\Ac_{i+1}(T)$. Then, $\Ac_1(T)=\{a_1,a_n\}$, $\Ac_2(T)=\{a_1,a_2,a_{n-1},a_n\}$, and so on. Then, $n/2$ steps of the sequence are needed to cover $S'$. On the other hand, if $\beta < 1-\delta$, then 3 edges to $T$ are enough for a vertex to become a part of $\Ac_1(T)$, so $S'\subset \Ac_1(T)$. 

Continuing with the proof of Proposition \ref{prop:output_is_correct}, the lower bounds on the sequences $\{\Gc_i(T)\}_{i=0}^\infty$ and $\{\Gc_i(S)\}_{i=0}^\infty$ can be improved by introducing:

\begin{definition}\label{def:improved_lower_bound_on_decoder_sequence}
For $G=(A \cup B, E)$ a bipartite graph, $C,S \subset A$ and $D,T \subset B$, let 
\begin{align}
\Acu_C(T) &= \bigcup_{S'\subset C : \Bu(S') \subset T} S' \\
\Bcu_C(T) &= \B(\Acu_C(T)) \\
\Bcu_D(S) &= \bigcup_{T' \subset D: \Au(T') \subset S} T' \\
\Acu_D(S) &= \A(\Bcu_D(T))
\end{align}
\end{definition}

Then,

\begin{lemma}\label{lemma:improved_lower_bound_on_decoder_sequence}
Let $G=(A \cup B, E)$ be a $(\gamma_A,\gamma_B,\delta_A,\delta_B)$ two-sided unique neighbor expander. Let $\delta=\max(\delta_A,\delta_B)$, let $\beta \in [0,1-\delta]$ and let 
\begin{align}
\ell_A &= \min\left(|A|,\left\lceil\frac{\log|E|}{\log\frac{1-\beta}{\delta_A}}\right\rceil\right) \\
\ell_B &= \min\left(|B|,\left\lceil \frac{\log|E|}{\log\frac{1-\beta}{\delta_B}} \right\rceil\right)
\end{align}
Let $C,S \subset A$, $D,T \subset B$ be such that $|\E(C)|\leq\gamma_A|E|$, $|\E(D)|\leq\gamma_B|E|$. Then:
\begin{align}
\Acu_C(T) &\subset \Ac_{\ell_A}(T) \\
\Bcu_C(T) &\subset \Bc_{\ell_A}(T) \\
\Acu_D(S) &\subset \Ac_{\ell_B}(S) \\
\Bcu_D(S) &\subset \Bc_{\ell_B}(S)
\end{align}
\end{lemma}

\begin{proof}
Lemma \ref{lemma:unique_neighbors_of_union_and_difference} implies that $\Bu(\Acu_C(T))\subset T$. Moreover, $\Acu_C(T) \subset C$, so $|\E(\Acu_C(T))|\leq\gamma_A|E|$. Then, Lemma \ref{lemma:lower_bound_on_decoder_sequence} implies that $\Acu_C(T) \subset \Ac_{\ell_A}(T)$ and $\Bcu_C(T) \subset \Bc_{\ell_A}(T)$ follows by monotonicity (Lemma \ref{lemma:monotonicity_of_neighborhood}). The other two parts are similar. 
\end{proof}

This section concludes with a short Lemma about column support that is used to argue that $\cs(X_B'')$ must be a subset of $\Bcu_D(\cs(W))$:

\begin{lemma}\label{lemma:about_column_support}
Let $S\subset A$, $D\subset B$ and $X\in\F^{B\times B}$ be such that $\cs(X)$ is a subset of $D$ and $\cs(XH)$ is a subset of $S$. Then, $\cs(X)$ is a subset of $\Bcu_D(S)$.  
\end{lemma}

\begin{proof}
All unique neighbors of $\cs(X)$ are in $\cs(XH)\subset S$, and $\Bcu_D(S)$ is defined as the union of all subsets of $D$ whose unique neighbors are in $S$. 
\end{proof}

\subsubsection{Stabilizer equivalent errors}\label{sec:stabilizer_equivalent_errors}

The following Lemma is used to replace an error with a stabilizer equivalent one and at the same time to change the row support of the error:

\begin{lemma}\label{lemma:stabilizer_equivalent_error}
Let $G=(A \cup B, E)$ be a bipartite graph with adjacency matrix $H \in \F^{B \times A}$ and consider the hypergraph product of $G$ with itself. Let $S \subset A$ and $T \subset B$ be such that the submatrix $\subm{H}{T}{S}$ is invertible. Suppose error $(X_A,X_B)$ with $X_A \in \F^{A \times A}$, $X_B \in \F^{B\times B}$ is such that the rows of the syndrome $HX_A+X_BH$ indexed by $T$ are zero and such that $\A(T) \cap \rs(X_A) \subset S$. Consider the stabilizer equivalent error $(X_A',X_B')$ where
\begin{align}
X_A' &= X_A + (\esubmi{H}{S}{T}X_B) H\\
X_B' &= X_B + H(\esubmi{H}{S}{T}X_B)
\end{align}
Then,
\begin{align}
\rs(X_A') &= \rs(X_A)\backslash S \\
\cs(X_A') &\subset \cs(X_A) \\
\rs(X_B') &\subset \rs(X_B)\cup\B(S) \\
\cs(X_B') &\subset \cs(X_B)
\end{align}
\end{lemma}

\begin{proof}
The condition that the rows of the syndrome indexed by $T$ are zero is equivalent to
\begin{align}
0=\esubmi{H}{S}{T}HX_A+\esubmi{H}{S}{T}X_BH
\end{align}
so an equivalent expression for $X_A'$ is
\begin{align}
X_A'= (I+\esubmi{H}{S}{T}H) X_A
\end{align}
From here, $\cs(X_A')\subset\cs(X_A)$ follows immediately. Next, note that $X_A'$ has block decomposition
\begin{align}
&\twobytwo{0}{\submi{H}{S}{T}\subm{H}{T}{A\backslash S}}{0}{\subm{I}{A\backslash S}{A\backslash S}}\twobyone{\subm{X}{S}{A}}{\subm{X}{A\backslash S}{A}} \\
&= \twobyone{\esubmi{H}{S}{T}\subm{H}{T}{A\backslash S} \subm{X}{A\backslash S}{A}}{\subm{X}{A\backslash S}{A}} \\
&=\twobyone{0}{\subm{X}{A\backslash S}{A}}
\end{align}
where the last step follows from the assumption $\A(T)\cap\rs(X_A)\subset S$. 

Next, the column support of $X_B'=(I+H\esubmi{H}{S}{T})X_B$ is a subset of the column support of $X_B$. The row support of $X_B'$ is included in the union of $\rs(X_B)$ and $\rs(H\esubmi{H}{S}{T})$, and the latter is a subset of $\B(S)$. 
\end{proof}

Repeated application of Lemma \ref{lemma:stabilizer_equivalent_error} gives:

\begin{lemma}\label{lemma:row_support_included_in_lower_bound_set}
Let $G=(A \cup B, E)$ be a bipartite graph with adjacency matrix $H \in \F^{B \times A}$ and consider the hypergraph product of $G$ with itself. Let $X_A \in \F^{A \times A}$, $X_B \in \F^{B\times B}$, $C \subset A$, $D \subset B$ be such that $(X_A,X_B)$ is $(C,D)$ contained (recall Definition \ref{def:cd_contained}). Let $T$ be the row support of the syndrome $HX_A + X_BH$. Then, there exist $X_A'\in\F^{A\times A}$, $X_B'\in\F^{B\times B}$ such that $(X_A',X_B')$ is stabilizer equivalent to $(X_A,X_B)$ and $(C,D)$ contained and such that the row support of $X_A'$ is a subset of $\Acu_C(T)$ (recall Definition \ref{def:improved_lower_bound_on_decoder_sequence}).
\end{lemma}

\begin{proof}
If $\rs(X_A)\subset\Acu_C(T)$, the claim holds. Otherwise, by the definition of $\Acu_C(T)$, $\Bu(\rs(X_A))$ is not a subset of $T$. Take $b \in \Bu(\rs(X_A))\backslash T$ and let the unique edge from $b$ to $\rs(X_A)$ be $(a,b)$. Lemma \ref{lemma:stabilizer_equivalent_error} applied to the submatrix $\subm{H}{b}{a}=1$ implies that there is $(X_A'',X_B'')$ that is stabilizer equivalent to $(X_A,X_B)$, and with $\rs(X_A'')=\rs(X_A)\backslash\{a\}$. Moreover, $(X_A'',X_B'')$ is also $(C,D)$ contained, because $\B(a)\subset \B(C)$. This step can be repeated until the remaining error in $\F^{A\times A}$ has row support included in $\Acu_C(T)$. 
\end{proof}

\subsubsection{Support of the syndrome remaining after the correction $\hat{X}_A$}\label{sec:support_of_remaining_syndrome}

Applying the inverse of a matrix is equivalent to solving a linear system and a variant of this observation involving submatrices is given in the following Lemma. It implies that certain rows of the remaining syndrome are zero. Later, it is also used to show that the corrections can be computed in linear time or logarithmic depth.  

\begin{lemma}\label{lemma:applying_the_inverse_of_a_submatrix_is_equivalent_to_solving_a_linear_system}
Let $H\in\F^{B\times A}$. Let $T\subset B$, $S \subset A$ be such that $\subm{H}{T}{S}$ is invertible. Let $P$ be another finite set, $W \in \F^{B\times P}$, $X \in \F^{A \times P}$. The following are equivalent:
\begin{enumerate}
\item $X=\esubmi{H}{S}{T} W $.
\item $\rs(X) \subset S$ and $\rs(W+HX)\subset B\backslash T$. 
\end{enumerate}

Similarly, let $Q$ be another finite set, $U \in \F^{Q \times A}$, $Y \in \F^{Q \times B}$. The following are equivalent:
\begin{enumerate}
\item $Y=U\esubmi{H}{S}{T}$.
\item $\cs(Y) \subset T$ and $\cs(U+YH)\subset A \backslash S$.
\end{enumerate}
\end{lemma}

\begin{proof}
First, assume $X=\esubmi{H}{S}{T}W$. Then, $\rs(X) \subset \rs(\esubmi{H}{S}{T}) = S$. Moreover, 
\begin{align}
\rs(W+HX)\subset\rs(I+H\esubmi{H}{S}{T}) = B\backslash T
\end{align}
where the second step follows from the block decomposition
\begin{align}
I+H\esubmi{H}{S}{T} = \twobytwo{\subm{0}{T}{T}}{\subm{0}{T}{B\backslash T}}{\subm{H}{B\backslash T}{S}\submi{H}{S}{T}}{\subm{I}{B\backslash T}{B\backslash T}}
\end{align}

Next, assume $\rs(X)\subset S$ and $\rs(W+HX) \subset B \backslash T$. From the first deduce $\esubmi{H}{S}{T}HX=X$ and from the second $\esubmi{H}{S}{T}(W+HX)=0$. The two conclusions imply $X=\esubmi{H}{S}{T}W$. 

The proof of the second part is similar. 
\end{proof}

From here follows: 

\begin{lemma}\label{lemma:support_of_remaining_syndrome}
The syndrome $W+H\hat{X}_A$ remaining after the correction $\hat{X}_A$ has
\begin{align}
\rs(W+H\hat{X}_A) &\subset \Bc_\ell(\rs(W))\backslash\Bp(\Ac_\ell(\rs(W)))\\
\cs(W+H\hat{X}_A) &\subset \cs(W)
\end{align}
\end{lemma}

\begin{proof}
First, $\rs(\hat{X}_A) \subset \Ac_\ell(\rs(W))$, so $\rs(H\hat{X}_A)\subset \Bc_\ell(\rs(W))$; in addition, $\rs(W) \subset \Bc_\ell(\rs(W))$, so $\rs(W+H\hat{X}_A)$ is a subset of $\Bc_\ell(\rs(W))$. Next, note that Lemma \ref{lemma:applying_the_inverse_of_a_submatrix_is_equivalent_to_solving_a_linear_system} implies that the rows of $W+H\hat{X}_A$ indexed by $\Bp(\Ac_\ell(\rs(W))$ are zero. Finally, $\cs(\hat{X}_A)\subset\cs(W)$ implies $\cs(W+H\hat{X}_A)\subset\cs(W)$. 
\end{proof}

\subsubsection{The second correction}\label{sec:the_second_correction}

The final step of the proof of Proposition \ref{prop:output_is_correct} is:

\begin{lemma}\label{lemma:b_side_correction_equals_remaining_error}
$\hat{X}_B=X_B''$.
\end{lemma}

\begin{proof}
\begin{align}
\hat{X}_B &=(H\hat{X}_A+W)\esubmi{H}{\Ap(\Bc_\ell(\cs(W)))}{\Bc_\ell(\cs(W))} \\
&= X_B''H\esubmi{H}{\Ap(\Bc_\ell(\cs(W)))}{\Bc_\ell(\cs(W))} \\
&= X_B'' \twobytwo{\subm{I}{\Bc_\ell(\cs(W))}{\Bc_\ell(\cs(W))}}{0}{H''}{0} \\
&=X_B''
\end{align}
where $H''=\subm{H}{B\backslash\Bc_\ell(\cs(W))}{\Ap(\Bc_\ell(\cs(W)))}\submi{H}{\Ap(\Bc_\ell(\cs(W)))}{\Bc_\ell(\cs(W))}$. The last step is because the column support of $X_B''$ is a subset of $\Bcu_D(\cs(W))$, which is included in $\Bc_\ell(\cs(W))$.
\end{proof}

\subsection{Proof of Proposition \ref{prop:linear_time}: linear time and logarithmic depth}\label{sec:proof_of_linear_time}

\subsubsection{Computing the row and column support of the syndrome}

\begin{lemma}\label{lemma:computation_of_the_row_and_column_support_of_the_syndrome}
The row and column support of any $W \in \F^{B\times A}$ can be computed in time $O(|A||B|)$ in the uniform cost model or by a circuit of depth $O(\log|A|+\log|B|)$ with $O(|A||B|)$ gates. 
\end{lemma}

\begin{proof}
The sequential algorithm: for $b \in B$, for $a\in A$, if $W_{b,a}=1$, add $a$ to $\cs(W)$ and add $b$ to $\rs(W)$. 

The parallel algorithm: whether a row of $W$ is not zero can be computed by a binary tree of OR gates, so the depth is $O(\log|A|)$ and the number of gates is $O(|A|)$ per row. Similarly, whether a column of $W$ is not zero can be computed by a binary tree of OR gates with depth $O(\log|B|)$ and with $O(|B|)$ gates per column. 
\end{proof}

\subsubsection{Finding rows and columns}

\begin{lemma}\label{lemma:computation_of_sequences}
Let $G=(A \cup B, E)$ be a bipartite graph with maximum left and right degrees $\dax,\dbx$ respectively. Let $\beta \in [0,1)$, $\ell \in \mathbb{N}$, $S \subset A$, $T \subset B$. Then, the sets $\Ac_\ell(T)$ and $\Bc_\ell(S)$ can be computed in time $O(|E|)$ in the uniform cost model or by a circuit with depth $O(\ell(\log\dax+\log\dbx))$ and with $O(\ell |E|)$ gates. 
\end{lemma}

\begin{proof}
The sequential algorithm to compute $\Ac_\ell(T)$ is:
\begin{enumerate}
\item Initialize $A':=\emptyset$, $B':=T$ (as arrays storing the respective indicator vector). Initialize array $M$ indexed by $A$. Initialize empty list $L_1$.  
\item For $a \in A$, $M(a):=|\B(a)\cap T|$, if $M(a)> \beta|\E(a)|$ add $a$ to the beginning of $L_1$. 
\item For $i$ from 1 to $\ell$
\begin{enumerate}
\item Create empty list $L_{i+1}$.
\item While $L_i$ is not empty
\begin{enumerate}
\item Let $a$ be the first element of $L_i$, remove $a$ from $L_i$ and add $a$ to $A'$. 
\item For each $b \in \B(a)\backslash B'$ add $b$ to $B'$, and for each $a\in\A(b)$ update $M(a):=M(a)+1$ and, moreover, if $M(a)-1 \leq \beta|\E(a)| < M(a)$, then add $a$ to the beginning of $L_{i+1}$. 
\end{enumerate}
\end{enumerate}
\item Output $A'$. 
\end{enumerate}
At the beginning of the $i$-th iteration, $M(a) = |\B(a) \cap \Bc_{i-1}(T)|$ and $L_i = \Ac_i(T)\backslash\Ac_{i-1}(T)$; the correctness of the algorithm follows. Each $a \in A$ can be added and removed from a list at most once; each edge is visited by the algorithm once during step 2 and at most twice during step 3, so the algorithm runs in time $O(|E|)$. 

The circuit to compute $\Ac_\ell(T)$ has $\ell$ layers. The $i$-th layer has input $|B|$ bits, indicating whether each $b \in B$ is an element of $\Bc_{i-1}(T)$. Then it takes the following steps:
\begin{enumerate}
\item For each $b \in B$, it uses a binary tree of fan-out gates to copy $|\E(b)|$ times the value $\I(b\in\Bc_{i-1}(T))$, where $\I$ denotes the indicator of a statement. This results in $|E|$ bits, indicating whether each $e \in E$ is incident to $\Bc_{i-1}(T)$; these are the input to Step 2. The original $|B|$ bits are an input to step 4.  
\item For each $a \in A$, it uses the circuits described in reference \cite[Sections III and IV]{parhami2009efficienthammingweight} to determine if more than $\beta|\E(a)|$ of the edges in $\E(a)$ are incident to $\Bc_{i-1}(T)$. This results in $|A|$ bits, indicating whether each $a \in A$ is an element of $\Ac_i(T)$. The last layer outputs these $|A|$ bits and stops; the other layers continue. 
\item For each $a \in A$, it copies $|\E(a)|-1$ times the value $\I(a\in\Ac_{i-1}(T))$. Combined with the output of Step 2, this gives $|E|$ bits indicating whether each $e \in E$ is incident to $\Ac_i(T)$. 
\item For each $b\in B$, it uses a binary tree of OR gates to determine whether $b \in \Bc_{i-1}(T)$ or there is some $e \in \E(b)$ that is incident to $\Ac_i(T)$. This results in $|B|$ bits, indicating whether each $b\in B$ is an element of $\Bc_i(T)$. The $i$-th layer outputs these bits, and they become the input of layer $i+1$. 
\end{enumerate}

The $i$-th layer has depth $O(\log\dax+\log\dbx))$ and has $O(|E|)$ gates. Step 1 has depth $O(\log\dbx)$ and has $O(|E|)$ gates, assuming a classical circuit model that counts fan-out as a separate gate. Step 2 has depth $O(\log\dax)$ and has $O(|E|)$ gates, because reference \cite{parhami2009efficienthammingweight} shows that the computation for each $a \in A$ can be done in depth $O(\log|\E(a)|)$ and with $O(|\E(a)|)$ gates. Step 3  has depth $O(\log\dax)$ and has $O(|E|)$ gates. Step 4 has depth $O(\log\dbx)$ and has $O(|E|)$ gates. 

The sequential and parallel algorithms to compute $\Bc_\ell(S)$ are similar. 
\end{proof}

\subsubsection{Computing peelings}

It is convenient to introduce the following sequences:

\begin{definition}\label{def:optimal_peeling_sequence}
Let $G=(A \cup B, E)$ be a bipartite graph. The optimal peeling sequence for $S \subset A$ is given by
\begin{align}
\Acp_0(S) &= S \\
\Acp_{i+1}(S) &= \Acp_i(S) \backslash \A(\Bu(\Acp_i(S)))
\end{align}
Similarly, the optimal peeling sequence for $T \subset B$ is given by 
\begin{align}
\Bcp_0(T) &=T \\
\Bcp_{i+1}(T) &= \Bcp_i(T) \backslash \B(\Au(\Bcp_i(T)))
\end{align}
\end{definition}

The sequences are optimal in the sense that at each stage, they remove all vertices that are connected to a unique neighbor. 

In a unique neighbor expander, the optimal peeling sequence of a small set quickly reaches the empty set:

\begin{lemma}\label{lemma:optimal_peeling_sequence_in_a_unique_neighbor_expander}
Let $G=(A \cup B, E)$ be a $(\gamma_A,\gamma_B,\delta_A,\delta_B)$ two-sided unique neighbor expander. Let 
\begin{align}
j_A &= \min\left(|A|,\left\lceil\frac{\log|E|}{\log\frac{1}{\delta_A}}\right\rceil\right)\\
j_B &= \min\left(|B|,\left\lceil\frac{\log|E|}{\log\frac{1}{\delta_B}}\right\rceil\right)
\end{align}
If $S\subset A$ has $|\E(S)|\leq \gamma_A|E|$ then $\Acp_{j_A}(S)=\emptyset$ and if $T\subset B$ has $|E(T)|\leq\gamma_B |E|$ then $\Bcp_{j_B}(T)=\emptyset$. 
\end{lemma}

\begin{proof}
From $|\Bu(\Acp_i(S))|>(1-\delta_A)|E(\Acp_i(S))|$ follows $|\E(\Acp_{i+1}(S))|<\delta_A|\E(\Acp_i(S))|$. The argument for $T$ is similar.  
\end{proof}

A peeling of a set can be obtained by computing its optimal peeling sequence, and simultaneously matching each vertex that is removed to a unique neighbor: 

\begin{lemma}\label{lemma:computation_of_peelings}
Let $G$ be a $(\gamma_A,\gamma_B,\delta_A,\delta_B)$ two-sided unique neighbor expander. Let
\begin{align}
j_A &= \min\left(|A|,\left\lceil\frac{\log|E|}{\log\frac{1}{\delta_A}}\right\rceil\right)\\
j_B &= \min\left(|B|,\left\lceil\frac{\log|E|}{\log\frac{1}{\delta_B}}\right\rceil\right)
\end{align}
If $S \subset A$ has $0<|\E(S)|\leq\gamma_A|E|$, then a peeling for $S$ exists and can be computed in time $O(|E|)$ in the uniform cost model or by a circuit with depth $O(j_A(\log\dax+\log\dbx))$, and with $O(|E|j_A\dax)$ gates. Similarly, a peeling for $T \subset B$ with $0< |E(T)| \leq \gamma_B|E|$ can be computed in time $O(|E|)$ or by a circuit with depth $O(j_B(\log\dax+\log\dbx))$ and with $O(|E|j_B\dbx)$ gates.
\end{lemma}

\begin{proof}
To compute a peeling for $S$ in $O(|E|)$ time:
\begin{enumerate}
\item Set $i:=1$, $j:=1$. Set $A':=S$ (as array storing the indicator vector). Initialize empty list $L_1$. Initialize array $M$ indexed by $B$. Initialize a $2$ by $|S|$ array that will store the output. 
\item For $b \in B$, $M(b)=|\A(b)\cap S|$, if $M(b)=1$ add $b$ to the beginning of $L_1$. 
\item While list $L_j$ is not empty
\begin{enumerate}
\item Create empty list $L_{j+1}$
\item Repeat
\begin{enumerate}
\item $b$ is the first element of $L_j$, remove $b$ from $L_j$. 
\item If $M(b)=1$, let $a$ be the only neighbor of $b$ in $A'$. Set the $i$-th edge of the matching to be $(a_i,b_i):=(a,b)$. Update $i:=i+1$. Remove $a$ from $A'$. For $b'\in\B(a)$, update $M(b'):=M(b')-1$, if $M(b')=1$ add $b'$ to the beginning of $L_{j+1}$. 
\end{enumerate}
until $L_j$ is empty
\item Update $j:=j+1$. 
\end{enumerate}
\item Output the peeling $(a_1,b_1),\dots,(a_{|S|},b_{|S|})$.
\end{enumerate}

At any time, $A'$ is the set of vertices in $S$ that have not yet been matched, and for $b\in B$, $M(b)$ is the number of edges that connect it to $A'$. In particular, $M(b)$ can only decrease, and the elements of $B$ with $M(b)=1$ are precisely $\Bu(S\backslash\{a_1,\dots,a_{i-1}\})$, which is not empty by unique neighbor expansion. 

The lists $L_j$ are used to quickly find a vertex that is a unique neighbor of $A'$. $b\in B$ is added to one of the lists when it becomes a unique neighbor of $A'$, and is removed either when it is matched to its only neighbor in $A'$ or some time after it stops being connected to $A'$. 

Several lists $L_j$ are used to make sure that the order of edges in the output respects the structure of the optimal peeling sequence. Specifically, at the start of iteration $j$ of the while loop, $A'=\Acp_{j-1}(S)$ and elements of the list $L_j$ such that $M(b)=1$ are precisely the vertices in $\Bu(\Acp_{j-1}(S))$. 

Each $a \in S$ and each $b \in B$ can become a part of the matching at most once, so the output is indeed a matching. Moreover, when $b_i$ is added, the only edge from it to $S\backslash\{a_1,\dots,a_{i-1}\}$ is to $a_i$, so the additional condition for a peeling also holds. Therefore, the output of the algorithm is indeed a peeling for $S$. 

Steps 1 and 2 take time $O(|E|)$, maintainence of the lists takes time $O(|B|)=O(|E|)$ (assuming no isolated vertices), maintainance of $M$ takes times $O(|E|)$, maintainance of $A'$ and of the output array takes time $O(|A|)$. Therefore, the whole algorithm takes time $O(|E|)$ in the uniform cost model.  

The circuit to compute a peeling of $S$ has $j_A$ layers; these suffice by Lemma \ref{lemma:optimal_peeling_sequence_in_a_unique_neighbor_expander}. The $i$-th layer has input $|A|$ bits that indicate if each $a \in A$ is an element of $\Acp_{i-1}(S)$. Then, 
\begin{enumerate}
\item For $a \in A$, it copies $2|\E(a)|$ times the bit $\I(a \in \Acp_{i-1}(S))$. The result is two groups of $|E|$ bits that indicate whether each $e \in E$ is incident to $\Acp_{i-1}(S)$. The first of these groups is an input to step 2, and the second to step 3. The original $|A|$ bits are an input to step 4. 
\item For each $b \in B$, it uses the methods of reference \cite{parhami2009efficienthammingweight} to determine if exactly one edge in $\E(b)$ is incident to $\Acp_{i-1}(S)$. The result is $|B|$ bits indicating whether each $b\in B$ is a unique neighbor of $\Acp_{i-1}(S)$. 
\item For each $b \in B$, it copies $|\E(b)|-1$ times the bit $\I(b\in\Bu(\Acp_{i-1}(S)))$. This results in $|E|$ bits indicating whether each $e \in E$ is incident to $\Bu(\Acp_{i-1}(S))$. A layer of AND gates with the second group of bits from step 1 results in $|E|$ bits indicating whether each $e \in E$ is an edge between $\Acp_{i-1}(S)$ and $\Bu(\Acp_{i-1}(S))$. 
\item For each $a \in A$, it determines whether $a \in \Acp_i(S)$. Moreover, if some edge in $\E(a)$ is an edge between $\Acp_{i-1}$ and $\Bu(\Acp_{i-1}(S))$, then it selects one such edge to become a part of the peeling. An explicit construction of a circuit for this computation is given in Lemma \ref{lemma:explicit_circuit} in Section \ref{sec:explicit_circuit}. This results in $|A|$ bits that indicate whether each $a \in A$ is an element of $\Acp_i(S)$ and this becomes the input to the next layer. Additionally, there are $|E|$ bits containing the indicator vector of the $i$-th layer of the peeling; these become the $i$-th part of the output. 
\end{enumerate}

Step 1 has depth $O(\log\dax)$ and has $O(|E|)$ gates. Step 2 has depth $O(\log\dbx)$ and has $O(|E|)$ gates. Step 3 has depth $O(\log\dbx)$ and has $O(|E|)$ gates. Step 4 has depth $O(\log\dax)$ and has $O(|E|\dax)$ gates because Lemma \ref{lemma:explicit_circuit} shows that the computation for each $a \in A$ can be done in depth $O(\log |\E(a)|)$ and with $O(|\E(a)|^2)$ gates. 

The sequential and parallel algorithms to compute a peeling of $T \subset B$ are similar. 
\end{proof}

\subsubsection{Computing the corrections}

The submatrix corresponding to a peeling is lower triangular, so its inverse can be applied by forward substitution. If the peeling respects the structure of the optimal peeling sequence, then blockwise forward substitution can be used to parallelize. 

\begin{lemma}\label{lemma:applying_the_inverse_of_a_submatrix}
Let $G=(A \cup B, E)$ be a $(\gamma_A,\gamma_B,\delta_A,\delta_B)$ two-sided unique neighbor expander with adjacency matrix $H \in \F^{B \times A}$ and with maximum left and right degrees $\dax$, $\dbx$. Let
\begin{align}
j_A &= \min \left( |A|. \left\lceil \frac{\log|E|}{\log\frac{1}{\delta_A}} \right\rceil \right)\\
j_B &= \min \left( |B|, \left\lceil \frac{\log|E|}{\log\frac{1}{\delta_B}} \right\rceil \right)
\end{align}
Let $S \subset A$, $0<|\E(S)|\leq\gamma_A|E|$, and let $(a_i,b_i)$, $i=1,\dots,|S|$ be the peeling of $S$ computed in Lemma \ref{lemma:computation_of_peelings}. Then, for every $W \in \F^{B \times A}$, the pair $ \hat{X}=\esubmi{H}{S}{\Bp(S)}W$ and $\hat{W}=W+H\hat{X}$ can be computed in time $O(|E||A|)$ in the uniform cost model or by a circuit with depth $O(j_A(\log\dax+\log\dbx))$ and with $O(|E||A|j_A)$ gates. Similarly, if $T \subset B$, $0<|\E(T)|\leq\gamma_B|E|$ and the peeling of $T$ is computed using Lemma \ref{lemma:computation_of_peelings}, then for every $W \in \F^{B \times A}$, the pair $\hat{X}=W\esubmi{H}{\Ap(T)}{T}$ and $\hat{W}=W+\hat{X}H$ can be computed in time $O(|E||B|)$ or by a circuit with depth $O(j_B(\log\dax+\log\dbx))$ and with $O(|E||B|j_B)$ gates. 
\end{lemma}

\begin{proof}
The sequential algorithm is forward substitution: 
\begin{enumerate}
\item Initialize $\hat{X}:=0$, $\hat{W}:=W$. 
\item For $i=1,\dots,|S|$:
\begin{enumerate}
\item Set $\subm{\hat{X}}{a_i}{A}:=\subm{\hat{W}}{b_i}{A}$. 
\item For $b\in \B(a_i)$, set $\subm{\hat{W}}{b}{A}:=\subm{\hat{W}}{b}{A}+\subm{\hat{X}}{a_i}{A}$. 
\end{enumerate}
\end{enumerate}
Each row update takes time $O(|A|)$, so the total number of operations is $O(|E||A|)$. At the beginning and after each execution of the loop, $\hat{W}=W+H\hat{X}$ holds. After the $j$-th execution of the loop, $\subm{\hat{W}}{b_j}{A}=0$ holds. This is not disturbed by any subsequent step, because $b_j$ is not connected to $a_i$ for any $i>j$. At the end $\rs(\hat{X}) \subset S$ and $\rs(\hat{W})\subset B\backslash \Bp(S)$, so Lemma \ref{lemma:applying_the_inverse_of_a_submatrix_is_equivalent_to_solving_a_linear_system} implies that $\hat{X}=\esubmi{H}{S}{\Bp(S)}W$. 

The parallel algorithm is blockwise forward substitution; indeed, the submatrix $\subm{H}{\Bp(S)}{S}$ is block lower triangular with identity blocks along the diagonal. In order to be able to refer to various submatrices, it is helpful to introduce shorthand notation for the set of vertices that are removed at the $i$-th step of the optimal peeling sequence:
\begin{align}
S_i=\Acp_{i-1}(S) \backslash \Acp_i(S)
\end{align}
and shorthand notation for the set of vertices that are removed in the first $i$ steps:
\begin{align}
S_i'=S \backslash \Acp_i(S) = \cup_{i'=1}^i S_{i'}
\end{align}
It is also helpful to be able to refer to the subset of $\Bp(S)$ that corresponds to $S_i$ or to $S_i'$. Therefore, for each $S' \subset S$, let $\Bp(S',S)$ denote the subset of $\Bp(S)$ that corresponds to $S'$. Thus, for each $i$, the submatrix $\subm{H}{\Bp(S_i',S)}{S_i'}$ is blockwise lower triangular and for each $i$ the diagonal block $\subm{H}{\Bp(S_i,S)}{S_i}$ is the identity. 

It is also helpful to introduce the notation $\X{0}=0$, $\W{0}_0=W$, and
\begin{align}
\X{i} &= \X{i-1}+\esubmi{H}{S_i}{\Bp(S_i,S)} \W{i-1} \\
\W{i} &= \W{i-1} + H (\X{i}+\X{i-1}) \\
&=W + H\X{i}
\end{align}
for $i=1, \dots, j_A$.

The parallel algorithm for computing the corrections has input $|B||A|$ bits representing $W$ and $j_A$ groups of $|E|$ bits, representing the indicator vector of the matching between $S_i$ and $\Bp(S_i,S)$ for $i=1,\dots,j_A$; these $j_A$ groups are the outputs of the parallel algorithm in Lemma \ref{lemma:computation_of_peelings}. The circuit for computing the corrections consists of pre-processing and $j_A$ layers. The pre-processing copies $|A|-1$ times each group of $|E|$ bits; using binary trees of fan-out gates, this takes depth $O(\log|A|)$ and $O(|E||A|j_A)$ gates. The $i$-th layer has input $|B||A|$ bits representing $\hat{W}_{i-1}$, $|A|^2$ bits representing $\hat{X}_{i-1}$, and $|E||A|$ bits representing $|A|$ copies of the indicator vector of the matching between $S_i$ and $\Bp(S_i,S)$. Then:
\begin{enumerate}
\item For each $b \in B$ and $a \in A$, copy $|\E(b)|$ times bit $\subm{\W{i-1}}{b}{a}$. This results in $|E||A|$ bits that are an input to step 2. The original $|B||A|$ bits of $\W{i-1}$ are an input to step 6. 
\item Apply a layer of AND gates between the $|E||A|$ bits of the previous step and the $|A|$ copies of the matching between $S_i$ and $\Bp(S_i,S)$. This results in $|E||A|$ bits such that for each $(a,b) \in E$, if $(a,b)$ is not part of the matching then the corresponding $|A|$ bits are zero, and if $(a,b)$ is part of matching then the corresponding $|A|$ bits are a copy of the row $\subm{\W{i-1}}{b}{A}$. 
\item For each $a,a' \in A$, take the OR of the bits from the previous step indexed by $((a,b),a')$ for $(a,b) \in \E(a)$. This results in $|A|^2$ bits representing $\X{i}+\X{i-1}$. 
\item Apply a layer of CNOT gates with controls the bits representing $\X{i}+\X{i-1}$ and targets the bits representing $\X{i-1}$. Thus, one register still holds $\X{i}+\X{i-1}$, and the other now holds $\X{i}$, which becomes an input to the next layer. 
\item For each $a,a' \in A$, copy $|\E(a)|-1$ times the bit $\subm{\X{i}}{a}{a'}+\subm{\X{i-1}}{a}{a'}$. This results in $|E||A|$ bits, representing, for each $(b,a)\in E$ and $a' \in A$, the intermediate result $\subm{H}{b}{a}(\subm{\X{i}}{a}{a'}+\subm{\X{i-1}}{a}{a'})$ of the matrix mutiplication $H(\X{i}+\X{i-1})$. 
\item For each $b \in B$, $a' \in A$, XOR the bit $\subm{\W{i-1}}{b}{a'}$ and the bits $\subm{H}{b}{a}(\subm{\X{i}}{a}{a'}+\subm{\X{i-1}}{a}{a'})$ for $(b,a) \in \E(b)$. This results in $|B||A|$ bits representing $\W{i}$, which become an input to the next layer. 
\end{enumerate}

Step 1 has depth $O(\log\dbx)$ and has $O(|A||E|)$ gates. Step 2 has depth $O(1)$ and has $O(|E||A|)$ gates. Step 3 has depth $O(\log\dax)$ and has $O(|A||E|)$ gates. Step 4 has depth $O(1)$ and has $O(|A|^2)$ gates. Step 5 has depth $O(\log\dax)$ and has $O(|E||A|)$ gates. Step 6 has depth $O(\log\dbx)$ and has $O(|E||A|)$ gates. 

Induction on $i$ shows that $\rs(\X{i}) \subset S_i'$ and that the rows of $\W{i}$ indexed by $\Bp(S_i',S)$ are zero; the inductive step uses the fact that the vertices in $\Bp(S_i,S)$ are a subset of $\Bu(\Acp_{i-1}(S))$ so layers after the $i$-th do not disturb rows indexed by $\Bp(S_i,S)$ of the remaining syndrome. Then, Lemma \ref{lemma:applying_the_inverse_of_a_submatrix_is_equivalent_to_solving_a_linear_system} implies that
\begin{align}
\X{i}=\esubmi{H}{S_i'}{\Bp(S_i',S)} W
\end{align}
Applying this for $i=j_A$ proves the correctness of the parallel algorithm, because Lemma \ref{lemma:optimal_peeling_sequence_in_a_unique_neighbor_expander} implies $S'_{j_A}=S$. 

The sequential and parallel algorithms to compute the correction involving the peeling of $T$ are similar, but operate with columns instead of rows. 
\end{proof}

\subsubsection{Explicit circuit construction for the proof of Lemma \ref{lemma:computation_of_peelings}}\label{sec:explicit_circuit}

It will be helpful to consider the following first:

\begin{definition}\label{def:cumulative_or}
A $d$ input cumulative OR circuit has $d$ input bits $\mu_1,\dots,\mu_d$ and outputs the $d$ bits
\begin{align}
\nu_1 &= \mu_1 \\
\nu_2 &= \mu_1 \lor \mu_2 \\
\nu_3 &= \mu_1 \lor \mu_2 \lor \mu_3 \\
& \dots \\
\nu_d &= \mu_1 \lor \mu_2 \lor \dots \lor \mu_d
\end{align}
\end{definition}

In a classical circuit model with bounded fan-in, a $d$ input cumulative OR circuit must have $\Omega(\log d)$ depth and $\Omega(d)$ gates to output the last bit which depends on all input bits. The following Lemma gives three constructions of cumulative OR: the first matches the lower bound on the depth, the second matches the lower bound on the number of gates, and the third comes close to matching both bounds at the same time. 

\begin{lemma}\label{lemma:cumulative_or_constructions}
There exists a $d$ input cumulative OR circuit with
\begin{enumerate}
\item depth $O(\log d)$ and with $O(d^2)$ gates, 
\item depth $\Theta(d)$ and with $O(d)$ gates,
\item depth $O((\log d)^2)$ and with $O(d\log d)$ gates. 
\end{enumerate}
\end{lemma}

\begin{proof}
Part 1: copy $d-i$ times bit $\mu_i$; using binary trees of fan-out gates, this takes depth $O(\log d)$ and has $O(d^2)$ gates. Then, use binary trees of OR gates to compute the bits $\nu_1,\dots,\nu_d$. This has depth $O(\log d)$ and has $O(d^2)$ gates. 

Part 2: copy bit $\mu_1$ once, compute $\mu_1\lor\mu_2$, copy this once, compute $(\mu_1\lor\mu_2)\lor\mu_3$, and so on. 

Part 3: apply a $d/2$ input cumulative OR to $\mu_1, \dots, \mu_{\frac{d}{2}}$ and to $\mu_{\frac{d}{2}+1},\dots, \mu_d$. Copy $d/2$ times the bit $\mu_1\lor\dots\lor\mu_{\frac{d}{2}}$, this adds $\log\left(\frac{d}{2}+1\right)$ to the depth and $d/2$ to the number of gates. Apply a layer of OR gates to the $d/2$ copies and the outputs of the second $d/2$ input cumulative OR; this adds 1 to the depth and $d/2$ to the number of gates. Then, the depth and the number of gates satisfy the recursions
\begin{align}
f(d) &= f\left(\frac{d}{2}\right) + \log\left(\frac{d}{2}+1\right)+1\\
g(d) &= 2 g\left(\frac{d}{2}\right) + d 
\end{align}
so the depth is $O((\log d)^2)$ and the number of gates is $O(d\log d)$. 
\end{proof}

Now, consider step 4 of the $i$-th layer of the parallel algorithm for the computation of a peeling in Lemma \ref{lemma:computation_of_peelings}. Consider this from the perspective of a node $a \in A$ of degree $d$. It has as input a bit $\chi$ indicating whether it is in $\Acp_{i-1}(S)$ and $d$ bits $\mu_1,\dots,\mu_d$ indicating whether each of its edges connects $\Acp_{i-1}(S)$ and $\Bu(\Acp_{i-1}(S))$. Node $a$ needs to determine if it is in $\Acp_i(S)$ and if necessary to select an edge to become a part of the peeling. The computational task of node $a$ can be formalized as follows:

\begin{definition}\label{def:edge_selection_circuit}
A degree $d$ edge selection circuit has $d+1$ input bits $\chi$, $\mu_1,\dots,\mu_d$ and transforms them to the $d+1$ output bits
\begin{align}
\chi' &= \chi \land (\lnot(\mu_1\lor\dots\lor\mu_d))\\
\nu_1 &= \mu_1\\
\nu_2 &= \mu_2\land(\lnot\mu_1)\\
\nu_3 &= \mu_3\land(\lnot(\mu_1\lor\mu_2))\\
&\dots \\
\nu_d &= \mu_d\land(\lnot(\mu_1\lor\dots\lor\mu_{d-1}))
\end{align}
\end{definition}

A degree $d$ edge selection circuit can be constructed from one $d$ input cumulative OR circuit with overhead of constant depth and linear number of gates. Combined with Lemma \ref{lemma:cumulative_or_constructions}, this gives three options for the computation performed by node $a$. For the proof of Lemma \ref{lemma:computation_of_peelings}, the one that minimizes the depth is used; optimizing the number of gates here is less important because the total number of gates of Find Rows and Columns and Decode is dominated by the computation of the corrections. 

\begin{lemma}\label{lemma:explicit_circuit}
There exists a degree $d$ edge selection circuit with
\begin{enumerate}
\item depth $O(\log d)$ and with $O(d^2)$ gates,
\item depth $\Theta(d)$ and with $O(d)$ gates,
\item depth $O((\log d)^2)$ and with $O(d \log d)$ gates. 
\end{enumerate}
\end{lemma}

\begin{proof}
Copy each bit $\mu_1,\dots,\mu_d$ once; this has depth 1 and $d$ gates. Apply a cumulative OR to the $d$ copies. Apply a layer of NOT gates to the outputs of the cumulative OR; this has depth 1 and $d$ gates. Combine the results with bits $\chi$, $\mu_1,\dots,\mu_d$; this has depth 1 and has $d$ AND gates. The three options for the overall complexity come from the three options for the cumulative OR from Lemma \ref{lemma:cumulative_or_constructions}
\end{proof}

\section{Conclusion}\label{sec:conclusion}

This article presented the Find Rows and Columns and Decode algorithm for quantum expander codes and established results about its complexity and decoding radius. The algorithm runs in linear time or logarithmic depth and has decoding radius that is up to a $\frac{1-3\delta}{1-2\delta}$ fraction of the expansion radius of the graph.  The analysis applies to expander graphs with $\delta < 1/3$ and with any degree distribution. 

One interesting question is if it is possible to improve the number of errors that can be efficiently corrected. When $\delta<1/4$, the ReShape reduction combined with the classical decoders of \cite{chen2023improveddecoding} corrects in quadratic time more errors than Find Rows and Columns and Decode. The challenge is to obtain an improvement in linear time or for $\delta \geq 1/4$ or both. 

Another interesting question is if Find Rows and Columns and Decode can be adapted to correct random errors occurring at a constant rate. The corresponding result for Small Set Flip uses the locality of that algorithm. The execution of Find Rows and Columns and Decode remains bounded as shown in Proposition \ref{prop:execution_is_contained}, so perhaps similar techniques can be used. 

A third direction for future work is to adapt the Find Rows and Columns and Decode algorithm to other quantum LDPC families, for example the asymptotically good codes of \cite{panteleev2021asymptotically,leverrier2022quantum,lin2022good} or the finite blocklength codes of \cite{panteleev2021degeneratequantum,ostrev2024quantum,kasai2026breakingorthogonalitybarrierquantum}. Results about adversarial or random errors would both be interesting. In the case of classical LDPC codes against random errors, it is possible to obtain performance that is close to the finite blocklength bounds \cite[Section IV-D and Figure 12]{polyanskiy2010channel}. Can an adaptation of Find Rows and Columns and Decode combined with some quantum LDCP code achieve performance close to the finite blocklength hashing bound \cite{ashikhmin2014fidelitylowerbounds,ostrev2026canonicalformandfiniteblocklengthbounds}?

\section*{Acknowledgement} 

This research was funded by the Luxembourg National Research Fund (FNR), grant reference C24/IS/18981686/QLDPC.

\end{document}